\documentclass[11pt]{article}

\usepackage{amsmath,amsthm,amssymb,amsfonts}
\usepackage{graphicx}
\usepackage[colorlinks,bookmarksopen,bookmarksnumbered,citecolor=red,urlcolor=red]{hyperref}
\usepackage{siunitx}
\usepackage[dvipsnames]{xcolor}
\usepackage{multirow}
\usepackage{booktabs}
\usepackage{authblk}
\usepackage{fullpage}

\newcommand{\R}{\mathbb{R}}
\newcommand{\X}{\mathbb{X}}
\newcommand{\U}{\mathbb{U}}
\newcommand{\Lc}{\mathcal{L}}
\newcommand{\Bc}{\mathcal{B}}

\DeclareMathOperator{\diag}{diag}

\newtheorem{remark}{Remark}

\newtheorem{definition}{Definition}
\newtheorem{theorem}{Theorem}

\title{Nonlinear parameter-varying state-feedback design for a gyroscope using virtual control contraction metrics}

\author[1]{Ruigang Wang}
\author[2]{Patrick J.W. Koelwijn}
\author[1]{Ian R. Manchester}
\author[2,3]{Roland T\'{o}th}
\affil[1]{Australian Centre for Field Robotics \& Sydney Institute for Robotics and Intelligent Systems, The University of Sydney, Sydney, NSW 2006, Australia}
\affil[2]{Department of Electrical Engineering, Eindhoven University of Technology, Eindhoven, The Netherlands}
\affil[3]{Systems and Control Laboratory, Institute for Computer Science and Control, Budapest, Hungary}

\begin{document}
	
\maketitle

\begin{abstract}	
In this paper, we present a virtual control contraction metric (VCCM) based nonlinear parameter-varying (NPV) approach to design a state-feedback controller for a control moment gyroscope (CMG) to track a user-defined trajectory set. This VCCM based nonlinear stabilization and performance synthesis approach, which is similar to linear parameter-varying (LPV) control approaches, allows to achieve exact guarantees of exponential stability and $\mathcal{L}_2$-gain performance on nonlinear systems with respect to all trajectories from the predetermined set, which is not the case with the conventional LPV methods. Simulation and experimental studies conducted in both fully- and under-actuated operating modes of the CMG show effectiveness of this approach compared to standard LPV control methods.	
\end{abstract}

\section{Introduction}

With increasing performance expectations and growing complexity of engineered systems, industrial control practice faces with achieving stabilization and shaping of the behavior of nonlinear dynamical systems. To address these problems, one possible methodology, which has seen rapid growth over the last few decades with many successful applications, is the so-called \emph{linear parameter-varying} (LPV) approach. In the LPV framework, the behavior, i.e., solution set, of a nonlinear system is embedded in an LPV representation, which has a linear dynamic relation between its inputs and outputs \cite{Toth2010SpringerBook}. This linear relation is dependent on a so-called \emph{scheduling variable}, a function of the states, inputs and/or outputs, that represents the nonlinear dynamical aspects of the original nonlinear system. The scheduling variable is assumed to be measurable in the system. This idea has allowed the successful extension of many analysis and synthesis tools of the \emph{linear time-invariant} (LTI) framework, such as the $\mathcal{L}_2$-gain stability and performance concept \cite{Apkarian:1995,Scherer1995,Wu:2006,Xie:2018}, to provide convex analysis and controller synthesis for nonlinear systems through the LPV framework. More recently, extensions have  been made to so-called \emph{nonlinear parameter-varying} (NPV) systems, where some nonlinear dynamics are still included in the model allowing for a less conservative representation of the nonlinear system \cite{Cai2015,Sala2019,Rotondoi:2019}. However, convex analysis and synthesis results are more difficult to obtain, due to the system not being linear, as is the case in the LPV framework.

While, the analysis and synthesis results of the LPV framework have successfully been applied in many engineering problems \cite{Mohammadpour2012,Hoffmann:2015}, recent research has shown that naively applying these results to nonlinear systems can result in incorrect analysis conclusions or unwanted closed-loop behavior in case of synthesis \cite{Scorletti:2015,Koelewijn:2019,Koelewijn:2020}. These issues stem from the fact that unlike LTI systems, stability properties of the origin and other forced equilibria are not equivalent for nonlinear systems. As a consequence, stability guarantees of an LPV/NPV embedding for the origin extend to that of the respective nonlinear model, but such guarantees are not sufficient to imply stability of all forced equilibria of the represented nonlinear system \cite{Koelewijn:2020}. 

As it turns out, the loss of guarantees are attributed to the used equilibrium-dependent stability notion -- widely applied in LPV control -- raising the question if with a different equilibrium-free stability concept such problems could be avoided without losing the convexity and attractive properties of LPV approaches. As an alternative, the concept of \emph{universal stabilization} aims to achieve exponential stability of all trajectories of the system \cite{Manchester:2017}. By a so-called \emph{control contraction metric} (CCM), analogous to the \emph{control Lyapunov function} (CLF) for a single equilibrium (or trajectory) \cite{Sontag:1983}, convex conditions can be derived for analysis and synthesis under universal stabilization \cite{Manchester:2018}. The ideas behind of these methods build on the concept of contraction analysis \cite{Lohmiller:1998,Forni:2014}, where analysis of convergence of the infinitesimal variations of the system around all trajectories (i.e., local stability of all trajectories) is equivalent with universal stability of the system (i.e., global stability of all trajectories). This leads to analysis and synthesis problems for a family of local linear systems, called \emph{differential dynamics}, that can be elegantly expressed as an LPV system and solved by LPV synthesis tools to give exact stability and performance guarantees on the resulting closed-loop nonlinear system through the CCM approach.

While the use of CCM allows to achieve universal stability and performance with LPV control, it may be a too strict notion if stabilization of only a particular subset of reference trajectories is required, as is common in tracking control. Hence, the notion of $\mathcal{B}^*$-\emph{universal stabilizability} has been introduced where stability of a subset of trajectories, denoted by $\mathcal{B}^*$, is aimed at, and which can be analyzed through so-called \emph{virtual control contraction metrics} (VCCMs) \cite{wang2020virtual}. The concept of VCCMs combines the notion of virtual systems \cite{Wang:2005} and CCMs. An earlier work of virtual contraction theory in control design can be found in \cite{Jouffroy:2010}. Some recent works include control synthesis for a special case of mechanical systems \cite{Manchester:2018a} and further extension to port-Hamiltonian systems \cite{Reyes:2019b,Reyes2020}. The main idea of virtual systems is that a nonlinear system, which is not itself contracting, may have weaker stability properties that can be established via construction of an auxiliary (virtual) system which is contracting. Furthermore, the virtual system can be seen as an NPV embedding of the dynamics of the original system. Then, its differential dynamics can still be expressed as a local LPV system, which allows the use of convex LPV synthesis results through the CCM approach, but with extended feasibility due to the reduced conservativeness of the embedding. The VCCM based control approach can achieve $\mathcal{B}^*$-universal stabilization and $\mathcal{L}_2$-gain performance guarantees for nonlinear systems. In contrast to the reference-dependent variable-gain tracking control approaches for linear systems \cite{Wouw:2008,Loon:2017}, the VCCM approach can deal with nonlinear systems and yield controllers whose gain depends on  both states and references.

In this paper the VCCM based controller design is applied in order to achieve $\mathcal{B}^*$-universal stabilizability and performance shaping for a \emph{control moment gyroscope} (CMG). CMGs have been widely applied in attitude control of ships \cite{Perez:2009}, satellites \cite{Lappas:2005}, and the international space station \cite{Gurrisi:2010}. They represent a challenging nonlinear system and are often used for the demonstration of nonlinear control methods \cite{Reyhanoglu:2006,Abbas:2014}. Two control configurations (fully- and under-actuated modes) of the CMG are considered. Furthermore, the influence of the used LPV or NPV embedding for the CMG on the achieved controller performance is investigated by constructing controllers using both type of embeddings in the VCCM based controller design approach. The simulation and experimental studies show how the choice of embedding model and control realization affects the closed-loop tracking performance. This type of question is not well-addressed in the LPV literature. The comparison results show that the NPV approach can ensure closed-loop stability and performance for user-specified tracking tasks while the conventional LPV approach may not provide such guarantees. 

The paper is structured as follows. In Section \ref{sec:problem}, a formal problem formulation is given, along with a description of the considered dynamical model for the CMG. Section \ref{sec:methodology} describes the VCCM based controller design for NPV embeddings. Section \ref{sec:cmgcontroldesign} details the controller designs applied to the CMG. In Section \ref{sec:result}, a simulation study is presented for the CMG using the introduced controller design methods and the results are thoroughly analyzed both form the view point of stability and achieved performance. Finally, in Section \ref{sec:conclusion}, concluding remarks on the presented work are given.

\subsection*{Notation}
$ \R $ is the set of real numbers, while $ \R_+ $ is the set of non-negative reals. Let $ (x,y) $ denote the vector concatenation of $ x\in\R^n,\,y\in\R^m $, i.e., $ (x,y):=[x^\top\; y^\top]^\top\in\R^{n+m} $. $ \mathcal{L}_2 $ is the space of square-integrable vector signals on $ \R_+ $, i.e., $ \|x\|_2:=\sqrt{\int_{0}^{\infty}|x(t)|^2dt}<\infty $ where $ |\cdot| $ is the Euclidean norm. The causal truncation $ (\cdot)_T $ is defined by $ (x)_T(t):=x(t) $ for $ t\in[0,T] $ and 0 otherwise. $ \mathcal{L}_2^\mathrm{e} $ is the space of vector signals on $ \R_+ $ whose causal truncation belongs to $ \mathcal{L}_2 $. For a matrix $ A$, $ A\succ 0 $ or $ A\succeq 0$ means that $ A $ is positive definite or positive semi-definite. Similarly $ A\prec 0 $ or $ A\preceq 0 $ means that $ A $ is negative definite or negative semi-definite. A Riemannian metric is a smooth matrix function $M:\R^n\rightarrow \R^{n\times n}$ with $M(x)\succ 0$ for all $x\in\R^n$. A metric $M(x)$ is said to be uniformly-bounded if there exist $a_2\geq a_1>0$ such that $a_1I\succeq M(x)\succeq a_2I$ for all $x\in\R^n$. Let $\gamma(x_0,x_1)$ be the set of smooth paths connecting $x_0$ to $x_1$, that is, each $c\in\Gamma(x_0,x_1)$ is a smooth map $c:[0,1]\rightarrow \R^n$ with $c(0)=x_0$ and $c(1)=x_1$. Given a metric $M(x)$, a geodesic $\gamma $ is a (non-unique) minimum length path defined by $\gamma:=\arg\inf_{c\in\Gamma(x_0,x_1)}\mathcal{E}(c)$ where $\mathcal{E}(c):=\int_0^1 c_s^\top M(c(s))c_s ds$. If $M$ is independent of $x$, then $\gamma$ is the straight line $\gamma(s)=(1-s)x_0+sx_1$. 

\section{Problem Formulation}\label{sec:problem}

\begin{figure}[!bt]
	\centering
	\begin{tabular}{cc}
		\includegraphics[width=0.28\linewidth]{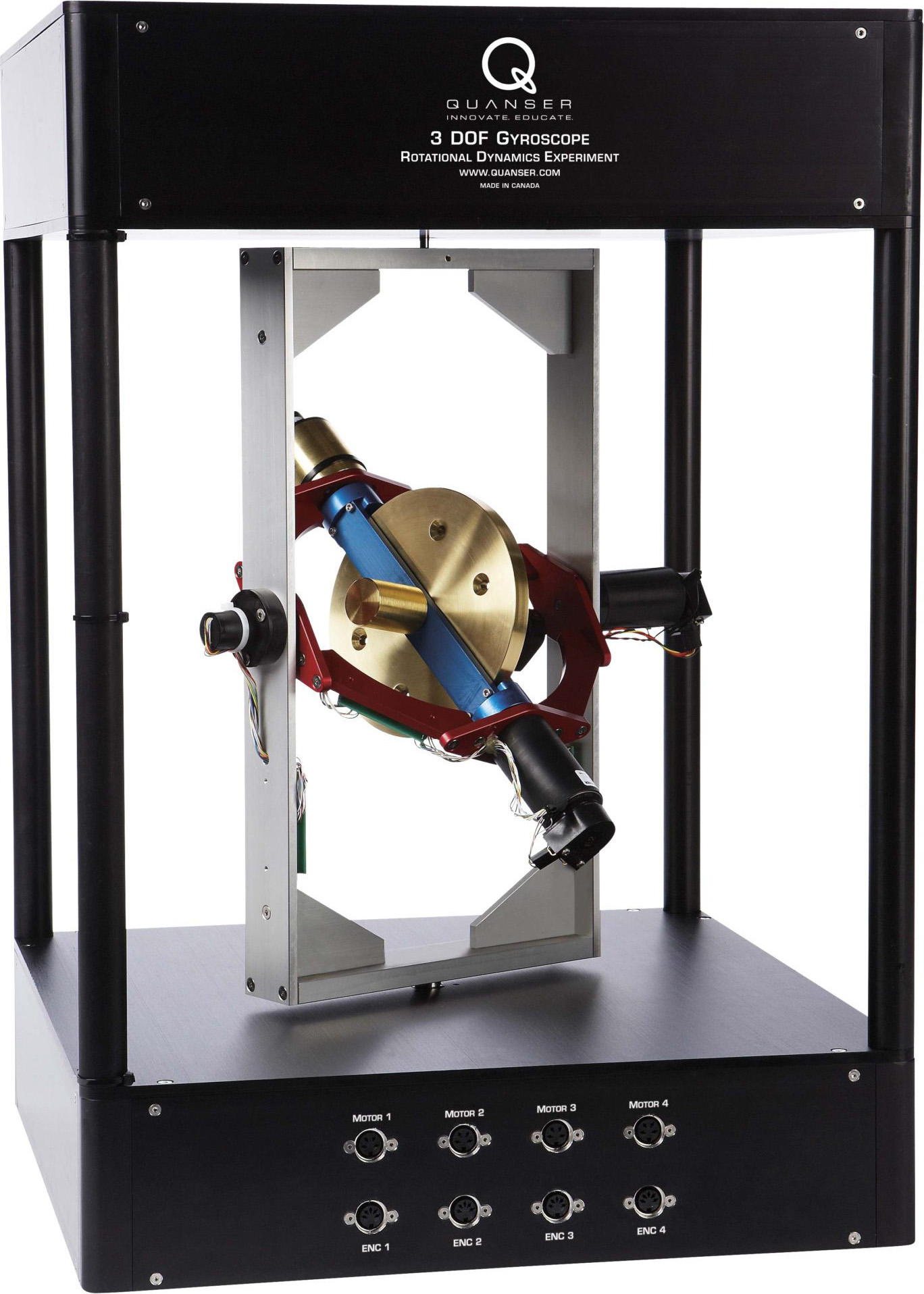} \hspace{0.05\linewidth} &
		\includegraphics[width=0.55\linewidth]{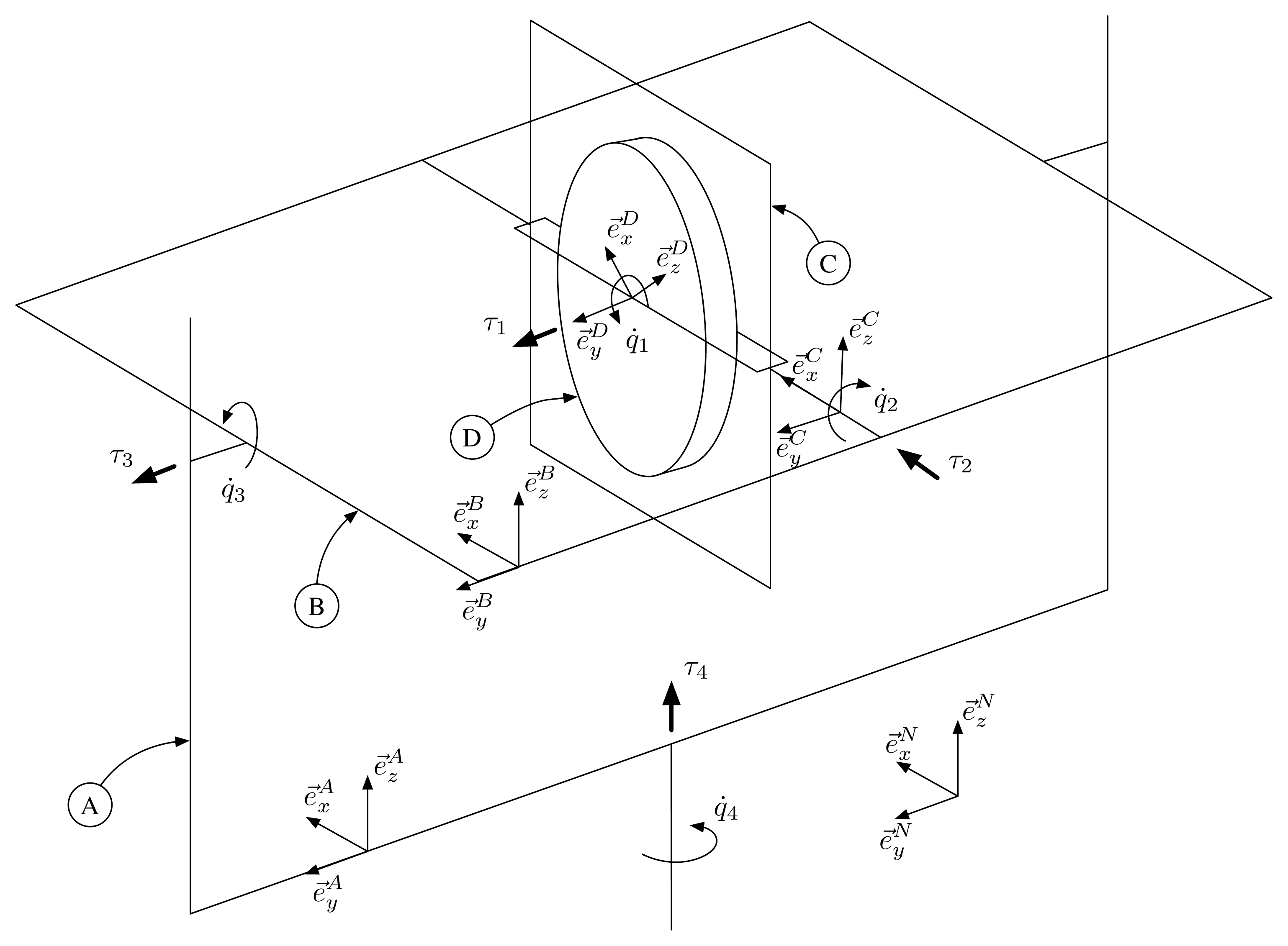} \\
		{\scriptsize (a) Setup } & {\scriptsize (b) Configuration}
	\end{tabular}
	\caption{The control moment gyroscope.}\label{fig:Gyro_Schematic}
\end{figure}

In this paper, we consider a 3-DOF CMG, see Figure \ref{fig:Gyro_Schematic}(a), consisting of three gimbals ($\mathrm{A}$, $\mathrm{B}$ and $\mathrm{C}$) along with a symmetric disk ($\mathrm{D}$), called the fly-wheel. The configuration of the CMG is depicted in Figure \ref{fig:Gyro_Schematic}(b). Let $ q=(q_1,q_2,q_3,q_4) $ be the generalized angular position vector and $ i=(i_1,i_2,i_3,i_4) $ be the motor currents vector. Here the index $ 1,2,3,4 $ refers to the frame $ \mathrm{D,C,B,A} $, respectively. The dynamics of the gyroscope can be represented by (\cite{Bloemers:2019,Berkel:2008})
\begin{equation}\label{eq:Lagrange-eq}
H(q)\ddot{q}+\bigl( C(q,\dot{q})+F_v \bigr) \dot{q}=K_mi,
\end{equation}
where $ F_v=\diag(f_v) $ with $ f_v $ as the viscous friction vector, $ K_m=\diag(k_m) $ with $ k_m $ as the motor constant vector. 
The inertia matrix is given as $ H(q) =\sum_{k \in \mathcal{S}} H_k(q_2,q_3) $ where the inertia matrices for each frame in $\mathcal{S}=\{\mathrm{A,B,C,D}\}$ are listed as follows:
\begin{align*}\allowdisplaybreaks[4]
	H_\mathrm{A} &= 
	\begin{bmatrix}
		0 & 0 & 0 & 0 \\
		\star & 0 & 0 & 0 \\
		\star & \star & 0 & 0 \\
		\star & \star & \star & K_\mathrm{A} 
	\end{bmatrix},\quad
	%Mb
	H_\mathrm{B} = 
	\begin{bmatrix}
		0 & 0 & 0 & 0 \\
		\star & 0 & 0 & 0 \\
		\star & \star & J_\mathrm{B} & 0 \\
		\star & \star & \star & I_\mathrm{B}s^2_3 + K_\mathrm{B}c^2_3 
	\end{bmatrix}, \\
	%Mc
	H_\mathrm{C} &= 
	\begin{bmatrix}
		0 & 0 & 0 & 0 \\
		\star & I_\mathrm{C} & 0 & -I_\mathrm{C}s_3 \\
		\star & \star & J_\mathrm{C}c^2_2+K_\mathrm{C}s^2_2 & \alpha_1s_2c_2c_3 \\
		\star & \star & \star & I_\mathrm{C}s^2_3+(J_\mathrm{C}s^2_2+K_\mathrm{C}c_2^2)c^2_3 
	\end{bmatrix}, \\
	%Md
	H_\mathrm{D} &= 
	\begin{bmatrix}
		J_\mathrm{D} & 0 & J_\mathrm{D}c_2 & J_\mathrm{D}s_2c_3 \\
		\star & I_\mathrm{D} & 0 & -I_\mathrm{D}s_3 \\
		\star & \star & I_\mathrm{D}s^2_2+J_\mathrm{D}c^2_2 & \alpha_2s_2c_2c_3 \\
		\star & \star & \star & I_\mathrm{D}s^2_3+(I_\mathrm{D}c^2_2+J_\mathrm{D}s^2_2)c^2_3 
	\end{bmatrix},
\end{align*}
with $ \alpha_1 = J_\mathrm{C}-K_\mathrm{C} $ and $ \alpha_2 = J_\mathrm{D}-I_\mathrm{D} $. For compactness and readability, sinusoidal functions are abbreviated as $s_i$ and $c_i$, e.g., $\sin q_2 = s_2$  and $\cos^2q_3 = c^2_3$. The terms $I_k$, $J_k$, $K_k$ with $k \in \mathcal{S}$ are the scalar moments of inertia about the $x$, $y$, $z$ axes respectively for the bodies $k$. The symbol $ \star $ denotes terms required to make the matrix symmetric. 

The elements of the Coriolis matrix $C(q,\dot{q})$ can be computed as:
\begin{align}
	C(q,\dot{q}) = \begin{bmatrix}
		\dot{q}^\top & 0 & 0 & 0 \\
		0 & \dot{q}^\top & 0 & 0 \\
		0 & 0 & \dot{q}^\top & 0 \\
		0 & 0 & 0 & \dot{q}^\top
	\end{bmatrix}
	\begin{bmatrix}
		\Gamma^1(q) \\ \Gamma^2(q) \\ \Gamma^3(q) \\ \Gamma^4(q)
	\end{bmatrix},
\end{align}
where
\begin{equation*}\allowdisplaybreaks[4]
	\begin{split}
		\Gamma^1 &= \frac{1}{2}
		\begin{bmatrix}
			0 & 0 & 0 & 0 \\
			\star & 0 & -J_\mathrm{D}s_2 & J_\mathrm{D}c_2c_3 \\
			\star & \star & 0 & -J_\mathrm{D}s_2s_3 \\
			\star & \star & \star & 0
		\end{bmatrix}, \quad
		% i = 2
		\Gamma^2 = \frac{1}{2}
		\begin{bmatrix}
			0 & 0 & J_\mathrm{D}s_2 & -J_\mathrm{D}c_2c_3 \\
			\star & 0 & 0 & 0 \\
			\star & \star & -2\alpha_3s_2c_2 & \alpha_3(c^2_2c_3-s^2_2c_3)-\alpha_4c_3 \\
			\star & \star & \star & \alpha_3c_2c^2_3s_2
		\end{bmatrix}, \\
		% i = 3
		\Gamma^3 &= \frac{1}{2}
		\begin{bmatrix}
			0 & -J_\mathrm{D}s_2 & 0 & J_\mathrm{D}s_2s_3 \\
			\star & 0 & 2\alpha_3s_2c_2 & \alpha_4c_3 + \alpha_3(c_3s^2_2-c^2_2c_3) \\
			\star & \star & 0 & 0 \\
			\star & \star & \star & -(\alpha_5+\alpha_3s^2_2)c_3s_3
		\end{bmatrix}, \\
		% i = 4
		\Gamma^4 &= \frac{1}{2}
		\begin{bmatrix}
			0 & J_\mathrm{D}c_2c_3 & -J_\mathrm{D}s_2s_3 & 0 \\
			\star & 0 & \alpha_3(c_3s^2_2-c^2_2c_3)-\alpha_4c_3 & -\alpha_3c_2c^2_3s_2 \\
			\star & \star & \alpha_3c_2s_2s_3 & (\alpha_5+\alpha_3s^2_2)c_3s_3 \\
			\star & \star & \star & 0
		\end{bmatrix},
	\end{split}
\end{equation*}
with $ \alpha_3 = I_\mathrm{D}-J_\mathrm{C}-J_\mathrm{D}+K_\mathrm{C} $, $ \alpha_4 = I_\mathrm{C}+I_\mathrm{D} $ and $ \alpha_5 = I_\mathrm{B}+I_\mathrm{C}-K_\mathrm{B}-K_\mathrm{C} $. The physical parameters of the gyroscope are given in the Table~\ref{tab:parameter}. Here we are interested in tracking control for the following two operating modes:
\begin{itemize}
	\item OM-1: Gimbal $ \mathrm{A} $ is locked, i.e. $ (q_4,\dot{q}_4)=0 $ and $ i_4=0 $. The control objective is to track a set of reference signals of $ \dot{q}_1,q_2 $ and $ q_3 $ using the input $ (i_1,i_2,i_3) $.
	\item OM-2: Gimbal $ \mathrm{B} $ is locked, i.e. $ (q_3,\dot{q}_3)=0 $ and $ i_3=0 $. The control objective is set-point tracking for $ \dot{q}_1 $ and $ q_4 $ by using the input $ (i_1,i_2) $. In this case, the motor on gimbal $\mathrm{A}$  is switched off, i.e. $ i_4=0 $, and hence the system is underactuated. 
\end{itemize}
Note that OM-1 is relatively easy to control as the CMG is fully-actuated. For OM-2, control design is a challenging task as the dynamics is highly nonlinear and under-actuated.

\begin{table}[!bt]
	\centering
	\caption{Model parameters of CMG.}\label{tab:parameter}
	\begin{tabular}{ccccccc}
		\toprule
		Index & \multicolumn{3}{c}{Moments} & Index & \multicolumn{2}{c}{Constants} \\ \midrule
		$ k $ & $ I $ & $ J $ & $ K $ & $ i $ & $ f_v $ & $ k_m $ \\ \midrule
		$ \mathrm{A} $ & \num{0.0902} & \num{0.0534} & \num{0.0374} & 1 & \num{1.1050e-5} & \num{0.0680} \\
		$ \mathrm{B} $ & \num{0.0039} & \num{0.0186} & \num{0.0200} & 2 & \num{1.2420e-5} & \num{0.1006} \\
		$ \mathrm{C} $ & \num{9.2087e-4} & \num{0.0016} & \num{0.0026} & 3 & \num{0.0141} & \num{0.1053} \\
		$ \mathrm{D} $ & \num{0.0030} & \num{0.0055} & \num{0.0374} & 4 & \num{0.0327} & \num{0.0606} \\ \bottomrule
	\end{tabular}
\end{table}

\section{Parameter-Varying Embeddings and Virtual Control Contraction Metrics}\label{sec:methodology}

\subsection{Control via NPV embedding}\label{sec:npv-control}
Consider nonlinear (NL) systems of the form
\begin{equation}\label{eq:system}
	\dot{x}=f(x,u),
\end{equation}
where $ x(t)\in\X\subseteq\R^n$ is the measured state and $ u(t)\in\U\subseteq\R^m$ is the control input. The function $ f $ is assumed to be sufficiently smooth. We define a \emph{reference trajectory} $ (x^*,u^*) $ to be a forward-complete solution of \eqref{eq:system}. A reference trajectory is said to be globally exponentially stabilizable if there exist a state feedback controller of the form
\begin{equation}\label{eq:controller}
	u=\kappa(x,x^*,u^*),
\end{equation}
where $ \kappa:\X\times\X\times\U\rightarrow\U $ such that the closed-loop (CL) system $ \dot{x}=f(x,\kappa(x,x^*,u^*)) $ is globally exponentially stable at $ (x^*,u^*) $, i.e., 
\begin{equation}
	|x(t)-x^*(t)|\leq Re^{-\lambda t}|x(0)-x^*(0)|,\quad \forall t>0,
\end{equation}
for some constants $ \lambda,R>0 $.

We will present a systematic approach to design controllers of the form \eqref{eq:controller} that achieve globally exponential stability for any reference trajectory $ (x^*,u^*) $ from a user-defined set $ \Bc^* $. If such controllers exist, we call system \eqref{eq:system} \emph{$ \Bc^* $-universally stabilizable}. Furthermore, if $ \Bc^* $ contains all reference trajectories, we simply call \eqref{eq:system} universally stabilizable. Note that depending on the choice of $ \Bc^* $, the task could be regulation, set-point tracking or reference tracking.  

In this work, we will first construct a \emph{virtual system} for \eqref{eq:system}, which is a new system of the form:
\begin{equation}\label{eq:sys-npv}
	\dot{\chi}=F(\chi,x,\mu),
\end{equation}
with the property of $ F(x,x,u)=f(x,u),\,\forall (x,u)\in\X\times\U $, where the virtual state $ \chi(t)\in\X$ and the virtual input $\mu(t)\in\U$ live in a copy of the state/input spaces of \eqref{eq:system}, and the external variable $ x(t) $ is taken as the state of \eqref{eq:system}. Note that the virtual system \eqref{eq:sys-npv} can also be understood as a nonlinear parameter-varying (NPV) embedding of \eqref{eq:system} since the behavior (solution set) of \eqref{eq:system} can be embedded into the behavior of \eqref{eq:sys-npv} via the map $ F $, which is called the \emph{behavior embedding principle}. The control design based on the behavior embedding principle usually includes three steps: the choice of a NPV model \eqref{eq:sys-npv},  the control synthesis based on it and the realization of the controller for the original system \eqref{eq:system}.

Note that the NPV embedding \eqref{eq:sys-npv} is not unique as there are various choices in terms of what level of nonlinearity is ``hidden'' in the external parameter. For example, the linear parameter-varying (LPV) form (\cite{Toth2010SpringerBook,Hoffmann:2015}) is an embedding where $ F $ is linear in $ \chi $ and $ \mu $. Furthermore, the system \eqref{eq:system} is a trivial embedding of itself when the full nonlinearity is considered. Note that an LPV embedding allows for simpler control synthesis, but may lead to conservative results. Compared to the linear (standard LPV) case, some recent works \cite{Rotondoi:2019} show that the performance can be improved by considering certain level of system nonlinearity. Here we construct the NPV embedding \eqref{eq:sys-npv} such that the following two conditions are satisfied:
\begin{itemize}
\item[{\bf C1})] For any trajectory $ x $ of system \eqref{eq:system}, the virtual system \eqref{eq:sys-npv} can be universally stabilized by a controller of the form
\begin{equation}\label{eq:control-fb}
	\mu=\mu^*+\kappa^{\mathrm{fb}}(\chi,\chi^*,x),
\end{equation}
with $ \kappa^{\mathrm{fb}}:\X\times\X\times\X\rightarrow\U $ and $ \kappa^{\mathrm{fb}}(\chi^*,\chi^*,x)=0,\ \forall \chi^*,x\in\mathbb{X} $, where $ (\chi^*,x,\mu^*) $ is an admissible trajectory of \eqref{eq:sys-npv}.
\item[{\bf C2})] There exists a controller $ \kappa^{\mathrm{ff}}:\X\times\X\times\U\rightarrow\U $ such that for any reference trajectory $ (x^*,u^*)\in \Bc^* $ and any trajectory $ x $ of system \eqref{eq:system}, $ (x^*,x,\mu^*) $ is a feasible solution to \eqref{eq:sys-npv}, where the feed-forward input $ \mu^* $ is given by
\begin{equation}\label{eq:control-ff-strong}
	\mu^*=\kappa^{\mathrm{ff}}(x,x^*,u^*).
\end{equation}
\end{itemize}

The following theorem gives a NPV controller that achieves $ \Bc^* $-universal stability for \eqref{eq:system}.
\begin{theorem}[\cite{wang2020virtual}]\label{thm:npv}
Consider the NL system \eqref{eq:system} and a reference set $ \Bc^* $. If there exists a NPV embedding \eqref{eq:sys-npv} such that Conditions {\bf C1} and {\bf C2} hold, then \eqref{eq:system} is $ \Bc^* $-universally stable under the controller
\begin{equation}\label{eq:control-output}
	u=\kappa^{\mathrm{ff}}(x,x^*,u^*)+\kappa^{\mathrm{fb}}(x,x^*,x). 
\end{equation}
\end{theorem}
\begin{proof}
Note that the above theorem is a special case of \cite[Thm. 2]{wang2020virtual}. Here we give a sketch proof as follows. By Condition {\bf C2} and the behavior embedding principle, the trajectories $(x^*,x,\mu^*)$ and $(x,x,u)$ with $u$ given in \eqref{eq:control-output} are two solutions of the CL virtual system of \eqref{eq:sys-npv} and \eqref{eq:control-fb}. Then, $ x(t) $ converges to $ x^*(t) $ exponentially as the CL system is contracting according to Condition \textbf{C1}.
\end{proof}
An equivalent nonlinear control framework has been proposed in \cite[Section 2.2.3]{Reyes:2019b}. The main difference is the control design method: \cite{Reyes:2019b} relays on exploring the physical structure (port-Hamiltonian) of the nonlinear system while the VCCM approach (see Section~\ref{sec:vccm-design}) focuses on constructive control design methods based on convex optimization. 

The NPV controller \eqref{eq:control-output} contains two parts: a feedback term that achieves universal stability for the NPV system and a feed-forward term that ensures any $ x^* $ from $ \Bc^* $ is admissible to the NPV embedding for all possible $ x $. Note that the universal stability notion used in Condition \textbf{C1} is much stronger than the standard stability concept. Under this strong notion, control synthesis for NPV systems can have a convex formulation similar to the LPV approach, as shown in the next section. The necessity of Condition {\bf C2} for CL stability guarantees will be discussed in Section~\ref{sec:NPV-ff}. 

\subsection{VCCM based control design}\label{sec:vccm-design}

This section presents a constructive approach \cite{wang2020virtual} to the universal stabilization problem in Condition {\bf C1}. For any fixed exogenous signal $ x $ generated by the original NL dynamics \eqref{eq:system}, the virtual system \eqref{eq:sys-npv} becomes a time-varying NL system, we can then apply the CCM-based method \cite{Manchester:2017} to design a universally stabilizing controller. In this approach, one considers a prolonged system consisting of \eqref{eq:sys-npv} and its \emph{differential dynamics}:
\begin{equation}\label{eq:diff-dyn}
	\dot{\delta}_\chi=A(\chi,x,\mu)\delta_\chi+B(\chi,x,\mu)\delta_\mu:=\frac{\partial F(\chi,x,\mu)}{\partial \chi}\delta_\chi+\frac{\partial F(\chi,x,\mu)}{\partial \mu}\delta_\mu,
\end{equation}
defined along solutions $ (\chi,\mu) $. Here $ (\delta_\chi,\delta_\mu) $ represents the infinitesimal variations between $ (\chi,\mu) $ and its neighborhood solutions. Note that we do not include variation on $x$ as it only needs to consider the contraction property of all virtual state trajectories $ \chi $ which are generated under the same exogenous signal $ x $. In this differential setting, many existing tools from the linear system theory (e.g. LMI based control design) can be applied.

A \emph{virtual control contraction metric} (VCCM) $ M(\chi,x) $ is a uniformly bounded matrix function $ M:\X\times\X\rightarrow \R^{n\times n} $ (i.e., there exist some $ a_2\geq a_1>0 $ such that $ a_1I\preceq M(\chi,x)\preceq a_2I $ for all $ \chi,x $) such that the following implication is true for all $ (\chi,x,\mu)\in\X\times\X\times\U $:
\begin{equation}\label{eq:ccm-def}
	\delta_\chi\neq 0,\,\delta_\chi^\top MB=0\;\Rightarrow\;\delta_\chi^\top(\dot{M}+A^\top M+MA+2\lambda M)\delta_\chi<0.
\end{equation}
The existence of a VCCM implies that \eqref{eq:sys-npv} is universally stabilizable \cite{Manchester:2017}. 
Furthermore, we can find a dual metric $ W=M^{-1} $ and a matrix function $ Y(\chi,x)\in\R^{m\times n} $ satisfying
\begin{equation}\label{eq:ccm-lmi}
	-\dot{W}+AW+WA^\top+BY+Y^\top B^\top+2\lambda W\preceq 0
\end{equation}
for all $ (\chi,x,\mu)\in\X\times\X\times\U $. An alternative way to find a CCM is to solve differential Riccati equations \cite{Schaft:2015}. Note that Condition~\eqref{eq:ccm-lmi} convex, but infinite dimensional, as the decision variables $ M,Y $ are smooth matrix functions. Finite-dimensional approximations include LPV synthesis techniques \cite{Toth2010SpringerBook} or sum-of-squares relaxation \cite{Parrilo:2003}. The pointwise LMI \eqref{eq:ccm-lmi} yields a differential state-feedback controller
\begin{equation}\label{eq:diff-control}
	\delta_\mu=K(\chi,x)\delta_\chi:=Y(\chi,x)W^{-1}(\chi,x)\delta_\chi.
\end{equation}

The realization task is to construct a controller satisfying Condition \textbf{C1} from the differential gain $K$. One solution is the path integral based realization \cite{Manchester:2017} of the local LPV controller \eqref{eq:diff-control}:
\begin{equation}\label{eq:control}
	\mu=\mu^*+\underset{\kappa^{\mathrm{fb}}(\chi,\chi^*,x)}{\underbrace{\int_{0}^{1}K(\gamma(s),x)\gamma_s(s)ds}}
\end{equation}	
where $ \gamma $ is a geodesic connecting $\chi^*$ to $\chi$ with respect to the metric $M(\chi,x)=W^{-1}(\chi,x)$. The above realization satisfies $ \kappa^{\mathrm{fb}}(\chi^*,\chi^*,x)=0,\ \forall \chi^*,x\in\mathbb{X} $ and $ \mu_s=K(\gamma(s),x)\gamma_s $, that is, it is an exact realization of \eqref{eq:diff-control} with $ \delta_\mu=\mu_s $ and $ \delta_\chi=\gamma_s $ along the path $ \gamma $. 

The realization \eqref{eq:control} has an LPV interpretation as follows. First, we take sufficiently many sample points of the path $ \gamma $ (i.e. $ 0=s_0<s_1<\cdots<s_N=1 $) such that $ \gamma(s_{i+1})-\gamma(s_i)\approx\gamma_s(s_i)\Delta_{s_i} $ where $ \Delta_{s_i}=s_{i+1}-s_i $, as shown in Figure~\ref{fig:path-integral}. Then, we can define a control sequence for those points by
\begin{equation}\label{eq:control-discrete}
	\nu(s_{i+1})= \nu(s_i)+K(\gamma(s_i),x)\gamma_s(s_i)\Delta_{s_i}=\nu(s_i)+K(\gamma(s_i),x) (\gamma(s_{i+1})-\gamma(s_i))
\end{equation}
with $ \nu(s_0)=\mu^* $. This can be understood as a series of local LPV controllers where each control action $ \nu(s_{i+1}) $ tries to make $ \gamma(s_{i+1}) $ exponentially converge to $ \gamma(s_i) $. Thus, it also makes $ \gamma(s_{i+1}) $ exponentially converge to the reference point $ \chi^*=\gamma(s_0) $. When the number of intermediate states approaches infinity (i.e. $\Delta_{s_i}\rightarrow 0$), the sequence \eqref{eq:control-discrete} becomes a smooth control path $ \nu:[0,1]\rightarrow\U $ defined as the path integral of the local LPV controller \eqref{eq:diff-control}:
\begin{equation}\label{eq:control-path}
	\nu(s):=\mu^*+\int_{0}^{s}K(\gamma(\mathfrak{s}),x)\gamma_s(\mathfrak{s})d\mathfrak{s}.
\end{equation}
The controller \eqref{eq:control} is the end point of this path, i.e. $ \mu=\nu(1) $. If $ M,K $ are independent of $ \chi $ and $ \mu $ respectively, we can obtain an explicit controller of the form
\begin{equation}\label{eq:realization-simple}
	\mu=\mu^*+\left[\int_{0}^{1}K(\widetilde{\chi}(s),x) ds\right] (\chi-\chi^*),
\end{equation}
where $ \widetilde{\chi}(s)=\chi^*+s(\chi-\chi^*) $. For the general case where $M$ is $\chi$-dependent, the controller \eqref{eq:control} usually requires solving an online optimization problem to construct a geodesic. For fast-sampling applications, there exist some real-time approximation methods, e.g. pseudo-spectral method \cite{Leung:2017} and gradient flows \cite{wang:2019}.

\begin{figure}[!bt]
	\centering
	\includegraphics[width=0.8\linewidth]{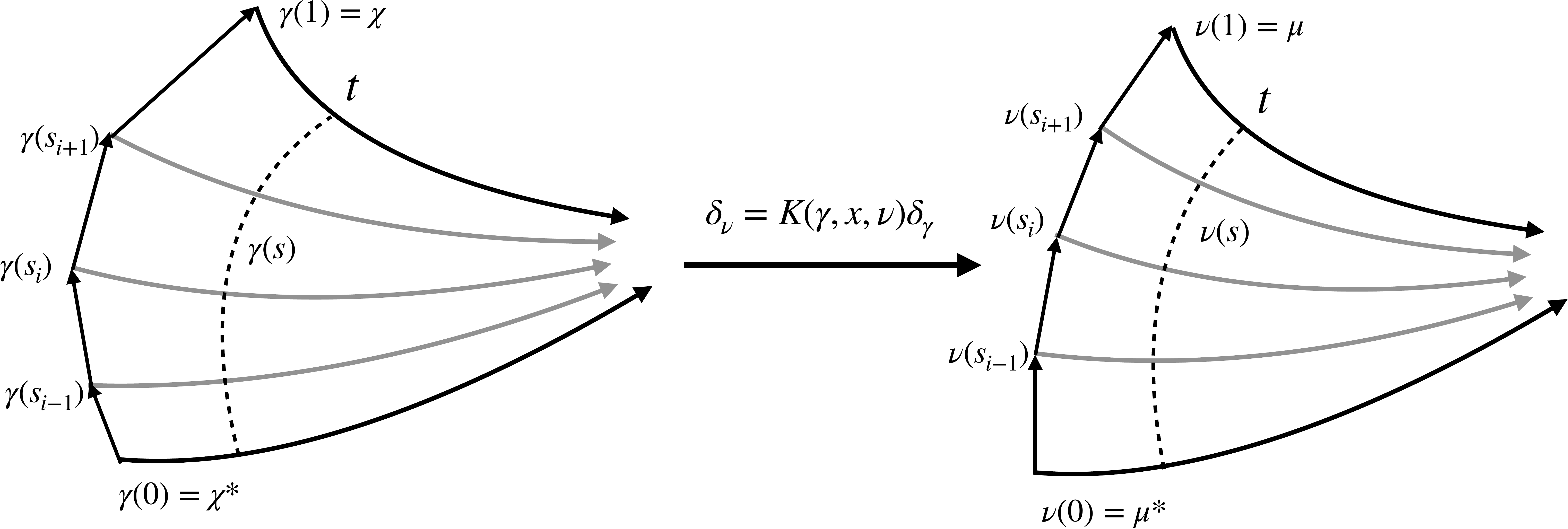}
	\caption{An LPV interpretation of the path integral based realization.}\label{fig:path-integral}
\end{figure}

\subsection{Performance design}

We will also consider the performance design for the $ \Bc^* $-specified tracking problem under load disturbance such as additive friction offsets, imbalance etc. As shown in Figure~\ref{fig:aug-sys}, the performance outputs of state error and control effort are defined by $z_1=W_1(x-x^*)$ and $z_2=W_2(u-u^*)$, respectively, where $W_1,W_2$ are stable linear weighting filters. With minor abuse of notation, we use $ x $ to refer to the state of the augmented system consisting of $ G,W_1$ and $W_2 $. 
\begin{figure}[!bt]
	\centering
	\includegraphics[width=0.55\linewidth]{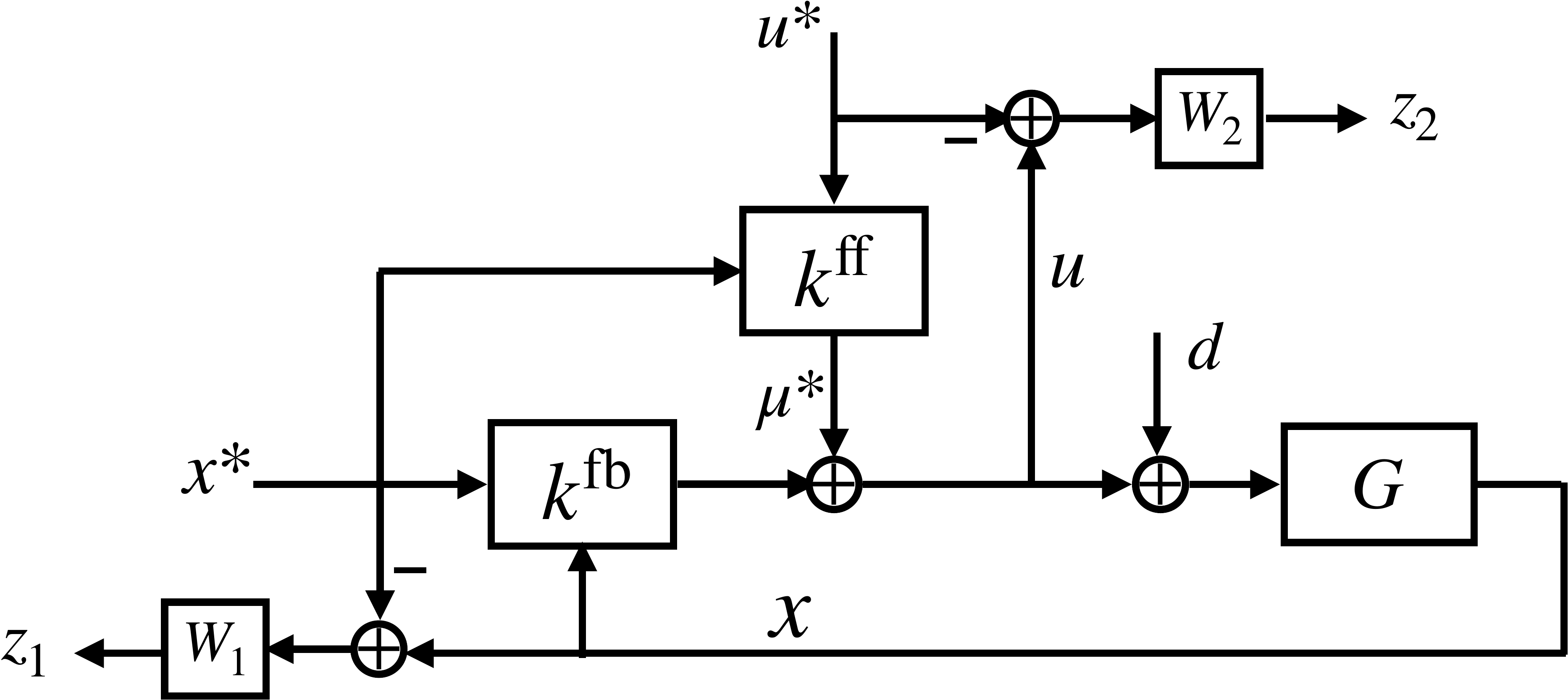}
	\caption{Diagram for the NPV based performance design. }
	\label{fig:aug-sys}
\end{figure}
The augmented system can be represented by the following general form
\begin{equation}\label{eq:sys-aug}
	\dot{x}=f(x,u,d),\quad z=h(x,u,d),
\end{equation}
where $z=(z_1,z_2)$. Applying the controller \eqref{eq:controller} to \eqref{eq:sys-aug} yields the CL system
\begin{equation}\label{eq:sys-aug-cl}
	\dot{x}=f(x,\kappa(x,x^*,u^*),d),\quad z=h(x,\kappa(x,x^*,u^*),d).
\end{equation}

\begin{definition}
The CL system \eqref{eq:sys-aug-cl} is said to achieve $ \Lc_2 $-gain bound of $ \alpha $ at $ (x^*,u^*,d^*,z^*) $ if for all $ T>0 $
\begin{equation}\label{eq:B-L2-gain}
	\|z-z^*\|_T^2\leq \alpha^2\|d-d^*\|_T^2+\beta(x(0),x^*(0)),
\end{equation}
for some function $ \beta(x_1,x_2)\geq 0 $ with $ \beta(x,x)=0 $, where $ d^*=0,\,z^*=0 $ are the nominal values of the disturbance and performance output, respectively. The controlled system \eqref{eq:sys-aug-cl} is said to have a $ \Bc^* $-universal $ \Lc_2 $-gain bound of $ \alpha $, if \eqref{eq:B-L2-gain} holds for all reference trajectories $ (x^*,u^*)\in\Bc^* $. If $ \Bc^* $ is the set of all feasible reference trajectories, we simply call \eqref{eq:sys-aug-cl} universal $ \Lc_2 $-gain bounded by $ \alpha $.
\end{definition}
Note that the $ \Bc^* $-universal gain condition \eqref{eq:B-L2-gain} is stronger than the standard $ \Lc_2 $-gain bound, but weaker than the incremental $ \Lc_2 $ gain bound \cite{Manchester:2018}. To extend the NPV approach for the disturbance rejection problem, we first construct a NPV virtual system of the form
\begin{equation}\label{eq:sys-aug-npv}
	\dot{\chi}=F(\chi,x,\mu,d),\quad \zeta=H(\chi,x,\mu,d),
\end{equation}
where $ \chi(t),\mu(t),\zeta(t) $ live in the same spaces as $ x(t),u(t),z(t) $, respectively. From the NPV embedding principle, we have $ F(x,x,u,d)=f(x,u,d) $ and $ H(x,x,u,d)=h(x,u,d) $. 

The associated differential dynamics of \eqref{eq:sys-aug-npv} is
\begin{equation}\label{eq:diff-sys-aug}
	\begin{split}
		\dot{\delta}_{\chi} &=A(\sigma)\delta_{\chi}+B(\sigma)\delta_{\mu}+B_d(\sigma)\delta_{d}, \\
		\delta_{\zeta} &= C(\sigma)\delta_{\chi}+D(\sigma)\delta_{\mu}+D_d(\sigma)\delta_{d},
	\end{split}
\end{equation}
where $ \sigma=(\chi,x,\mu,d)$, $ A=\frac{\partial F}{\partial \chi} $, $ B=\frac{\partial F}{\partial \mu} $, $ B_d=\frac{\partial F}{\partial d} $, $ C=\frac{\partial H}{\partial \chi} $, $ D=\frac{\partial H}{\partial \mu} $ and $ D_d=\frac{\partial H}{\partial d} $. Applying the differential state feedback \eqref{eq:diff-control} to \eqref{eq:diff-sys-aug} gives the CL differential dynamics:
\begin{equation}\label{eq:diff-sys-aug-cl}
	\begin{split}
		\dot{\delta}_{\chi}&=(A+BK)\delta_{\chi}+B_d\delta_{d}, \\
		\delta_{\zeta}&=(C+DK)\delta_{\chi}+D_d\delta_{d}.
	\end{split}
\end{equation}
To establish an $ \Lc_2 $-gain bound for \eqref{eq:diff-sys-aug-cl}, the choice of VCCM is a uniformly bounded metric $  M(\chi,x) $, satisfying
\begin{equation}\label{eq:diff-L2}
	\dot{V}(\chi,x,\delta_\chi)\leq -\frac{1}{\alpha}|\delta_\zeta|^2 +\alpha|\delta_d|^2,
\end{equation}
where $ V(\chi,x,\delta_\chi)=\delta_\chi^\top M(\chi,x)\delta_\chi $ can be interpreted as a differential storage function. Integration of the above dissipation condition along the geodesics gives the universal $ \Lc_2 $-gain bound of $ \alpha $ \cite{Manchester:2018}.

Similar to the case of $ H_\infty $  state-feedback control for linear systems (e.g., \cite{Dullerud:2013}), the condition \eqref{eq:diff-L2} can be converted into the following pointwise LMI:
\begin{equation}\label{eq:rvccm-synthsis}
	\begin{bmatrix}
		\mathcal{W} & B_d &(C W+D Y)^\top \\
		B_d^\top  & -\alpha I &   D_d^\top \\
		(CW+DY) &   D_d & -\alpha I
	\end{bmatrix}\prec 0,
\end{equation}
where $ W=M^{-1} $, $ Y=KW $ and $ \mathcal{W}=- \dot{W}+AW+WA^\top+BY+Y^\top B^\top $. Note that the above formulation is infinite dimensional but convex in $ W $ and $ Y $. The finite-dimensional approximation techniques for \eqref{eq:ccm-lmi} can also be applied here. 

From \cite[Th.~1]{Manchester:2018} and the behavior embedding principle, the CL system \eqref{eq:sys-aug-cl} achieves an $ \Lc_2 $-gain bound of $ \alpha $ from $ d-d^* $ to $ z-\zeta^* $ where $ \zeta^*=H(x^*,x,\kappa^{\mathrm{ff}}(x,x^*,u^*),d^*) $ satisfies $ \zeta^*=z^* $ if $ x=x^* $. Furthermore, if there exists a constant $ \alpha_{1}>0 $ such that $ |\zeta^*-z^*|\leq \alpha_{1}|x-x^*| $, we can obtain the $ \Lc_2 $-gain bound from $ d-d^* $ to $ \zeta^*-z^* $ as $ \alpha_{1}\alpha_{2} $ where $ \alpha_2 $ is the $ \Lc_2 $-gain bound from $ \delta_{d} $ to $ \delta_{\chi} $ of \eqref{eq:diff-sys-aug-cl} with $ K=YW^{-1} $. Then, the performance bound from $ d-d^* $ to $ z-z^* $ is given as follows.

\begin{theorem}[\cite{wang2020virtual}]\label{thm:npv-performance}
Consider the system \eqref{eq:sys-aug} and its NPV embedding \eqref{eq:sys-aug-npv}. Assume that the LMI \eqref{eq:rvccm-synthsis} is feasible and Condition \textbf{C2} holds for \eqref{eq:sys-aug-npv} and the reference set $ \Bc^* $. Then, the controller \eqref{eq:control-output} achieves a $ \Bc^* $-universal $ \Lc_2 $-gain bound of $ \widetilde{\alpha}=\sqrt{\alpha^2+(\alpha_{1}\alpha_{2})^2} $.
\end{theorem} 

\begin{remark}
When $ d=d^* $, the tracking cost $ J_T(x_0,x_0^*):=\int_{0}^{T}|z(t)|^2dt $ is bounded by
\begin{equation}\label{eq:perf-bound}
	J_T(x_0,x_0^*)\leq J_\infty(x_0,x_0^*)\leq \alpha^2 \mathcal{E}(\gamma),
\end{equation}
where $ \gamma $ is a geodesic joining $ x_0^* $ to $ x_0 $.
\end{remark}

\subsection{Comparison with the LPV embedding approach}\label{sec:NPV-ff}

In LPV based state-feedback control, system \eqref{eq:system} is rewritten into an LPV embedding of the form
\begin{equation}\label{eq:stand-lpv}
	\dot{x}=\widehat{A}(\sigma)x+\widehat{B}(\sigma)u,
\end{equation}
where $\sigma=\phi(x)$ is the scheduling variable such that $\widehat{A}(\phi(x))x+\widehat{B}(\phi(x))u=f(x,u)$. Note that this embedding can be understood as an LPV virtual system as follows:
\begin{equation}\label{eq:sys-lpv}
	\dot{\chi}=A(x)\chi+B(x)\mu
\end{equation}
with $A(x)=\widehat{A}(\phi(x))$ and $B(x)=\widehat{B}(\phi(x))$. Since the virtual system is linear in $\chi$ and $\mu$, we can use the VCCM based synthesis formulation \eqref{eq:ccm-lmi} to construct an LPV controller $\mu=K(x)\chi$ such that the CL system $\dot{\chi}=A_c(x)\chi:=\bigl( A(x)+B(x)K(x)\bigr) \chi$ is exponentially stable with respect to a Lyapunov function $V(\chi)=\chi^\top M(x)\chi$. 

The standard LPV realization for reference tracking takes the form of
\begin{equation}\label{eq:lpv-stand}
	u=u^*+K(x)(x-x^*),
\end{equation}
where $(x^*,u^*)$ is a feasible solution of \eqref{eq:system}. The VCCM approach uses a different realization, denoted as LPV-VCCM controller, which has the form of
\begin{equation}\label{eq:lpv-vccm}
	u=\kappa^{\mathrm{ff}}(x,x^*,u^*)+K(x)(x-x^*),
\end{equation}
where the feed-forward term $\kappa$ satisfies Condition \textbf{C2}, i.e., $\dot{x}^*=A(x)x^*+B(x)\kappa^{\mathrm{ff}}(x,x^*,u^*)$. Note that when $x^*$ is the origin, the standard LPV and LPV-VCCM controllers have the same realization as $u=K(x)x$. For general cases, they are not identical, resulting in different CL behaviors. Stability of the VCCM-LPV approach can be rigorously  guaranteed by Theorem~\ref{thm:npv} while the standard LPV controller may not offer such guarantees \cite{Scorletti:2015,Koelewijn:2019,Koelewijn:2020}. 

Here we give a brief explanation for the loss of stability guarantees of the standard LPV controller, see \cite[Section 5.2]{wang2020virtual} for details. From \eqref{eq:sys-lpv} - \eqref{eq:lpv-stand}, we can derive the error dynamics as follows
\begin{equation*}
	\begin{split}
		\dot{e}&=A(x)x+B(x)[u^*+K(x)(x-x^*)]-A(x^*)x^*-B(x^*)u^* \\
		&=A_c(x)e+[A(x)-A(x^*)]x^*+[B(x)-B(x^*)]u^*=[A_c(x)+\Delta(x,x^*,u^*)]e,
	\end{split}
\end{equation*}
where $e:=x-x^*$ and
\begin{equation}\label{eq:residual}
	\Delta(x,x^*,u^*)(x-x^*)=[A(x)-A(x^*)]x^*+[B(x)-B(x^*)]u^*.
\end{equation}
Note that the term $\Delta$ vanishes when $x^*$ is the origin, otherwise it is generally non-zero. Now we look into the time derivative of the Lyapunov function $V(e)=e^\top M(x)e$:
\begin{equation}\label{eq:perf-loss}
	\dot{V}(e)=e^\top\mathcal{Q}e +e^\top(\Delta^\top M+M\Delta)e,
\end{equation}
where $\mathcal{Q}=\dot{M}+MA_c+A_cM\preceq -2\lambda M$. When $ \Delta $ is sufficiently large, a smaller converge rate or even instability can be observed, see the academic example in \cite[Section 5.2.2]{wang2020virtual}. The above analysis also applies to the performance design where the bound \eqref{eq:perf-bound} may not hold if Condition {\bf C2} is not satisfied. The VCCM approach can provide stability and performance guarantees since the feed-forward term $ \kappa^{\mathrm{ff}} $ satisfying Condition \textbf{C2} also ensures $ \Delta(x,x^*,u^*)\equiv 0 $.

\section{Control Design for CMG}\label{sec:cmgcontroldesign}

\subsection{Fully-actuated operating mode: OM-1}

Since the outer-most gimbal $ \mathrm{A} $ is locked in this mode (i.e., $ (q_4,\dot{q}_4)=0 $ and $ i_4=0 $), we only need to take into account part of the CMG dynamics whose state and input are $ x=(q_2,q_3,\dot{q}_1,\dot{q}_2,\dot{q}_3) $ and $ u=(i_1,i_2,i_3) $, respectively. The fly-wheel angle $q_1$ can be ignored since it does not directly affects the dynamics of other states. Then, the state-space model of OM-1 can be written as follows
\begin{equation}\label{eq:om1-sys}
	\dot{x}=A(x_1,x_2)x+B(x_1)u:=\begin{bmatrix}
	0 & E \\
	0 & \mathcal{H}(x_1)^{-1}(\mathcal{C}(x_1,x_2)+\mathcal{F}_v)
	\end{bmatrix}\begin{bmatrix}
	x_1 \\ x_2
	\end{bmatrix}+\begin{bmatrix}
	0 \\
	\mathcal{H}(x_1)^{-1}\mathcal{K}_m
	\end{bmatrix}u,
\end{equation}
where $ x=(x_1,x_2) $, $ x_1=(q_2,q_3) $, $ x_2=(\dot{q}_1,\dot{q}_2,\dot{q}_3) $ and $ E=\begin{bmatrix}0 & I\end{bmatrix} $. Here $ \mathcal{H},\mathcal{C},\mathcal{F}_v,\mathcal{K}_m $ are constructed by eliminating the 4th row and column of the matrices $ H,C,F_v,K_m $ in \eqref{eq:Lagrange-eq}, respectively. Note that the above dynamics are fully-actuated. For performance design, we consider the following perturbed dynamics:
\begin{equation}
	\begin{split}
		\dot{x}=A(x_1,x_2)x+B(x_1)u+Dd,\quad
		z=\begin{bmatrix}
		W_1(x-x^*) \\ W_2(u-u^*)
		\end{bmatrix},
	\end{split}
\end{equation}
with $ D=\begin{bmatrix}
0 & I
\end{bmatrix}^\top\in\R^{5\times 3} $, where $ d(t)\in\R^3 $ is the input perturbation and the weighting matrices is chosen as $W_1=1,W_2=0.2$.  

Here we only present the details about Lyapunov design as the performance design has the same procedure except solving a different point-wise LMI.

\paragraph{Standard LPV control.}
We consider the following LPV virtual system of \eqref{eq:om1-sys}:
\begin{equation}\label{eq:om1-lpv}
	\dot{\chi}=A(x_1,x_2)\chi+B(x_1)\mu
\end{equation}
where the scheduling variables satisfy $ q_2,q_3\in [-\frac{\pi}{3},\frac{\pi}{3}] $, $ \dot{q}_1\in[30,60] $ and $ \dot{q}_2,\dot{q}_3\in[-1,1] $. For Lyapunov design, the pointwise LMI \eqref{eq:ccm-lmi} is solved by the grid-based method \cite{Wu:1995} with $ \lambda=0.5 $ and constant dual metric $ W $. To be specific, the grid based approach approximates the LPV embedding \eqref{eq:om1-lpv} as a state-space model array defined on a finite grid domain, as shown in Figure~\ref{fig:grid}. For each grid point $x^k$, there is a corresponding LTI system $(A(x^k),B(x^k))$ which describes the dynamics of \eqref{eq:om1-lpv} where the scheduling variable $x$ is held constant. For control of CMG, we use three grid points for each scheduling variable. Then, the pointwise LMI \eqref{eq:ccm-lmi} is approximated by 243 LMIs of grid-dependent matrices $Y^k\in\R^{5\times 3}$ with $k=1,2,\ldots,243$ and a dual metric $W\in\R^{5\times 5}$. The control synthesis problem is solved by YALMIP \cite{Lofberg:2004} with the solver SDPT3 \cite{Toh:1999}, which takes roughly $6.5\mathrm{s}$ on MacBook Pro with Intel Core i5, 8GB memory and Matlab 2020a. 

We determine the LPV control gain $K_{\mathrm{LPV}}(x)$ via linear interpolation of grid-dependent gain $K^k=Y^kW^{-1}$, $1\leq k\leq 243$. The tracking controller for the reference $(x^*,u^*)$ takes the form of
\begin{equation}\label{eq:om1-stand-lpv}
	\text{standard LPV:}\quad u=u^*+K_{\mathrm{LPV}}(x)(x-x^*).
\end{equation}
For the two conditions of Theorem~\ref{thm:npv}, the above realization only satisfies Condition \textbf{C1}.

\begin{figure}[!bt]
	\centering
	\includegraphics[width=0.35\textwidth]{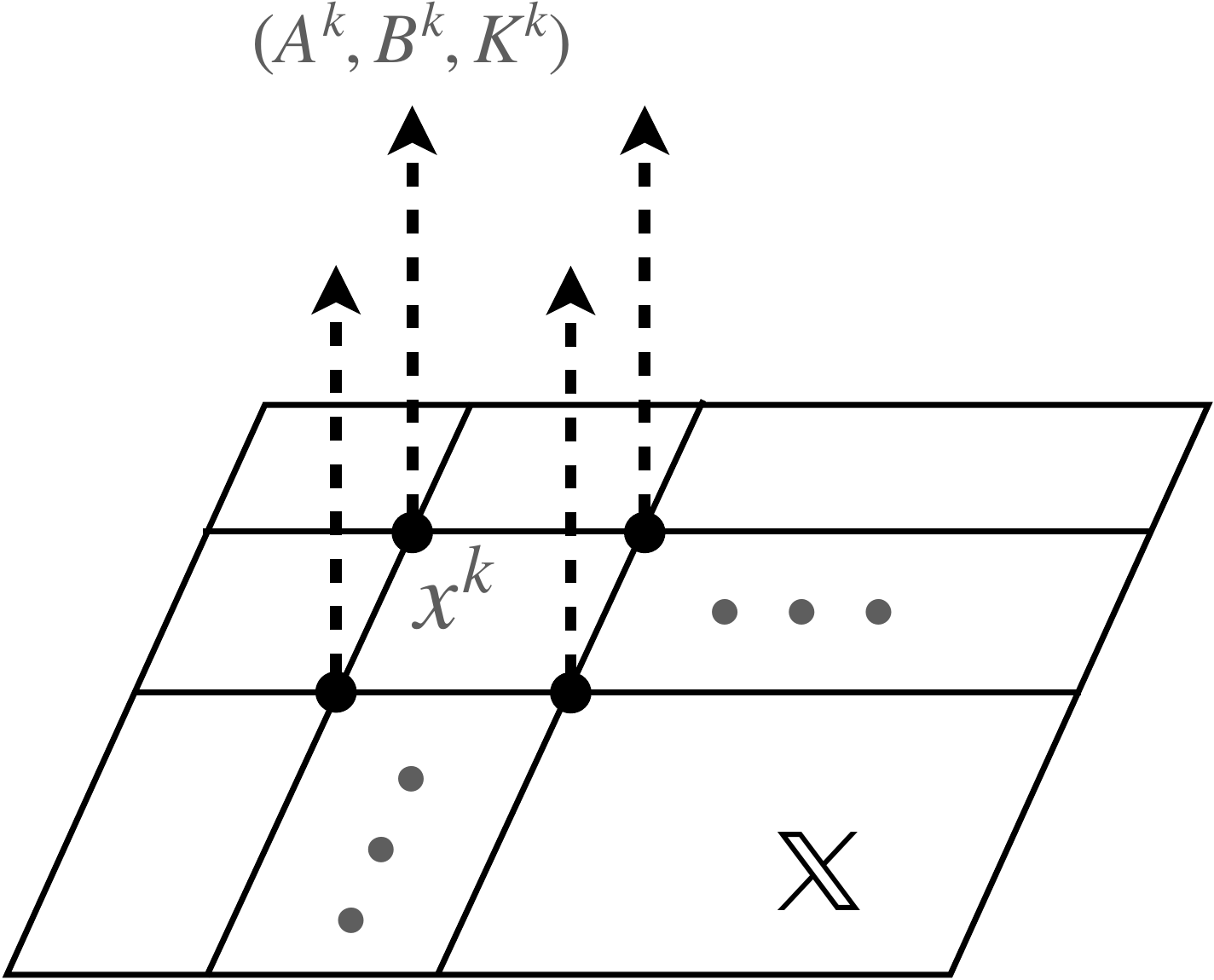} 
	\caption{LPV models and control gains defined on a rectangular grid.}\label{fig:grid}
\end{figure}

\paragraph{LPV-VCCM control.} We use the same design procedures of the standard LPV approach except the control realization. Here we choose the following controller
\begin{equation}\label{eq:om1-lpv-vccm}
	\text{LPV-VCCM:}\quad u=k_{\mathrm{LPV}}^{\mathrm{ff}}(x^*,u^*,x)+K_{\mathrm{LPV}}(x)(x-x^*), 
\end{equation}
where $ k_{\mathrm{LPV}}^{\mathrm{ff}}(x^*,u^*,x):=\mathcal{K}_m^{-1}[\mathcal{H}(x_1)\dot{x}_{2}^*+(\mathcal{C}(x_1,x_2)+\mathcal{F}_v)x_2^*] $. Note that Conditions \textbf{C1} - \textbf{C2} hold for the above LPV-VCCM controller. 

\paragraph{NPV-VCCM control.} We choose the following NPV embedding of \eqref{eq:om1-sys}:
\begin{equation}\label{eq:om1-npv}
	\dot{\chi}=A(x_1,\chi_2)\chi+B(x_1)\mu,
\end{equation}
which explicitly considers the quadratic nonlinearity of $ x_2 $ in the OM-1 dynamics. The nonlinearity of $ x_1 $ is hidden in the external parameter so that the $B$ matrix is independent of $\chi$, allowing a simpler formulation of \eqref{eq:ccm-lmi} as it becomes independent of $\mu$ \cite{Manchester:2017}. 

The associated differential dynamics of \eqref{eq:om1-npv} is
\begin{equation}
	\dot{\delta}_\chi=\mathcal{A}(\sigma)\delta_\chi+B(\sigma)\delta_\mu,
\end{equation}
where 
\[ \mathcal{A}(\sigma)=\begin{bmatrix}
0 & E \\
0 & \mathcal{H}(x_1)^{-1}(2\mathcal{C}(x_1,\chi_2)+\mathcal{F}_v)
\end{bmatrix}\] with $ \sigma=(x_1,\chi_2) $ as the scheduling variable. For control synthesis, the operating range of $\sigma$ is chosen to be the same as the LPV case. We solve \eqref{eq:ccm-lmi} via grid-based approach to obtain the differential control gain $ K_{\mathrm{NPV}}(x_1,\chi_2) $. The control realization takes the form as follows:
\begin{equation}\label{eq:om1-npv-vccm}
	\text{NPV-VCCM:}\quad u=k_{\mathrm{NPV}}^{\mathrm{ff}}(x,x^*,u^*)+\left[\int_{0}^{1}K_{\mathrm{NPV}}(x_1,\chi_2(s))ds\right](x-x^*)
\end{equation}
with $ \chi_2(s)=(1-s)x_2^*+sx_2 $, where the feed-forward term is chosen as
\begin{equation}
	k_\mathrm{NPV}^{\mathrm{ff}}(x,x^*,u^*):=\mathcal{K}_m^{-1}[\mathcal{H}(x_1)\dot{x}_2^*+(\mathcal{C}(x_1,x_2^*)+\mathcal{F}_v)x_2^*].
\end{equation}
Note that the above realization satisfies both Conditions \textbf{C1} and \textbf{C2}. For online computation, the integral in \eqref{eq:om1-npv-vccm} is approximated by
\begin{equation}
	\int_{0}^{1}K_{\mathrm{NPV}}(x_1,\chi_2(s))ds\approx \frac{1}{N}\sum_{i=0}^{N-1}K_{\mathrm{NPV}}(x_1,x_2^*+i(x_2-x_2^*)/N)
\end{equation}
where a large $ N $ can improve the accuracy, but results in online computation delay. Here we found that $ N=10 $ can provide a good accuracy with minor control latency in this case.

\subsection{Under-actuated operating mode: OM-2}\label{sec:cmg-om2}

For the operating mode OM-2, the CMG becomes an under-actuated system as there is no motor torque acting on gimbal $ \mathrm{A} $. The motion of gimbal $ \mathrm{A} $ is then mainly driven by the gyroscopic effect from disk $ \mathrm{D} $ and gimbal $ \mathrm{C} $. Specifically, when disk $ \mathrm{D} $ satisfies $\dot{q}_1>0$, the motion of gimbal $\mathrm{C}$ towards the direction $q_2>0$ will generate a torque to drive frame $ \mathrm{A} $ towards the direction $q_4<0$, and vice versa. Thus, the variables $ q_2 $ and $ q_4 $ cannot be independently controlled, e.g. we cannot move $ (q_2,q_4) $ to the region $ \R_+^2 $. To avoid the difficulties of constructing an embedding satisfying Condition C1 under such constraints, we exclude $q_2$ from the state value and treat it as a scheduling variable. The dynamics of OM-2 can be represented by
\begin{equation}\label{eq:om2-sys}
	\dot{x}=A(q_2,x_2)x+B(q_2)u:=
	\begin{bmatrix}
	0 & E \\
	0 & \mathcal{H}(q_2)^{-1}(\mathcal{C}(q_2,x_2)+\mathcal{F}_v)
	\end{bmatrix}\begin{bmatrix}
	x_1 \\ x_2
	\end{bmatrix}+\begin{bmatrix}
	0 \\
	\mathcal{H}(q_2)^{-1}\mathcal{K}_m
	\end{bmatrix}u,
\end{equation}
where $ x=(x_1,x_2) $ with $ x_1=q_4 $ and $ x_2=(\dot{q}_1,\dot{q}_2,\dot{q}_4) $ is the state, and $ u=(i_1,i_2) $ the control input. It is important to note that the scheduling variable $ q_2 $ is not a free external parameter as it is affected by the internal state $ \dot{q}_2 $ of \eqref{eq:om2-sys}. Here $ \mathcal{H},\mathcal{C},\mathcal{F}_v $ are constructed by eliminating the 3rd row and column of the matrices $ H,C,F_v $ in \eqref{eq:Lagrange-eq}, respectively, and $ \mathcal{K}_m $ is obtained by removing the 3rd row and 3-4th columns of $ K_m $. 

For performance design, we consider the following model:
\begin{equation}
	\begin{split}
		\dot{x}=A(q_2,x_2)x+B(q_2)u+Dd,\quad
		z=\begin{bmatrix}
		W_1(x-x^*) \\ W_2(u-u^*)
		\end{bmatrix},
	\end{split}
\end{equation}
with $ D=[0\;1\; 0\; 0]^\top $, where $ d(t)\in\R $ is input perturbation. The weighting matrices are chosen as $W_1=\diag(5,0.1,1,4)$ and $W_2=\diag(20,10)$, which will be explained later in Section~\ref{sec:result-om2}. Similar to control design for OM-1, the rest of this section focuses on the Lyapunov design.

\paragraph{Standard LPV control.} Here we consider the following LPV embedding of \eqref{eq:om2-sys}:
\begin{gather}\label{eq:om2-lpv} 
	\dot{\chi}=A(q_2,x_2)\chi+B(q_2)\mu, 
\end{gather}
where the scheduling variable $ (q_2,x_2) $ is chosen to be within the range of $ q_2\in [-\frac{\pi}{3},\frac{\pi}{3}] $, $ \dot{q}_1\in[30,60] $ and $ \dot{q}_2,\dot{q}_4\in[-1,1] $. We use the grid-based method to solve the pointwise LMI \eqref{eq:ccm-lmi} with $ \lambda=0.5 $ and constant dual metric $ W $. The control realization is similar to \eqref{eq:om1-stand-lpv}. 

\paragraph{LPV-VCCM control.} We use the same LPV embedding model and synthesis result from the standard LPV control design for OM-2. The corresponding LPV-VCCM controller can be written in the form of \eqref{eq:om1-lpv-vccm} but with a different feed-forward term
\begin{equation}\label{eq:om2-lpv-ff}
	k_{\mathrm{LPV}}^{\mathrm{ff}}(x,x^*,u^*):=\mathcal{K}_m^{\dagger}[\mathcal{C}(q_2,x_2)+\mathcal{F}_v]x_2^*,
\end{equation}
where $ \mathcal{K}_m^{\dagger}=(\mathcal{K}_m^\top \mathcal{K}_m)^{-1}\mathcal{K}_m^\top $ denotes the general inverse. Condition \textbf{C2} does not hold for the virtual system \eqref{eq:om2-lpv} as the term $\mathcal{K}_m\mu$ lives in a 2-dimensional space due to underactuation  while the vector $ \mathcal{C}(q_2,x_2)x_2^*$ is a full 3-dimensional vector when $ \dot{q}_2,\dot{q_4}$ are non-zeros.
%Since OM-2 is under-actuated, it is easy to verify that Condition \textbf{C2} is not satisfied for any feed-forward term. 
The choice in \eqref{eq:om2-lpv-ff} minimizes the residual term \eqref{eq:residual} so that CL performance loss as analyzed in Section~\ref{sec:NPV-ff} is reduced.

\paragraph{NPV-VCCM control.} We choose the trivial NPV embedding (i.e. the true system itself) of \eqref{eq:om2-sys}:
\begin{equation}\label{eq:om2-npv}
	\dot{\chi}=A(q_2,\chi_2)\chi+B(q_2)\mu. 
\end{equation}
The reason for such choice is that Condition \textbf{C2} can be easily satisfied by simply using the feed-forward term $\mu^*=u^*$. The NPV-VCCM control realization can be expressed as
\begin{equation}\label{eq:om2-npv-controller}
	u=u^*+\left[\int_{0}^{1}K_{\mathrm{NPV}}(q_2,\chi_2(s))ds\right](x-x^*),
\end{equation}
where $ \chi_2(s)=(1-s)x_2^*+sx_2 $. Here the control gain $K_{\mathrm{NPV}}(q_2,\chi_2)$ is obtained by solving \eqref{eq:ccm-lmi} with grid-based method subject to the same operation range as the standard LPV approach. 

\section{Discussions on simulation and experimental results}\label{sec:result}

\subsection{Operating mode OM-1}\label{sec:result-om1}

We first compare the standard LPV, LPV-VCCM, and NPV-VCCM controllers obtained from Lyapunov design. The control tasks include tracking of set-points and a dynamic reference. The simulation results are depicted in Figure~\ref{fig:om1-response} where the LPV-VCCM and NPV-VCCM controllers have similar CL convergence rate as they both satisfy Conditions \textbf{C1} and \textbf{C2}. The standard LPV controller has similar convergence rate for set-point tracking but fails to track the dynamic reference. As analyzed in Section~\ref{sec:NPV-ff}, this is mainly due to the violation of Condition \textbf{C2} for the standard LPV controller, which yields a residual term for the error dynamics as follows
\begin{equation}\label{eq:om1-residual}
	\Delta(x,x^*)(x-x^*)=[\mathcal{H}(x_1)-\mathcal{H}(x_1^*)]\dot{x}_2^*+[\mathcal{C}(x_1,x_2)-\mathcal{C}(x_1^*,x_2^*)]x_2^*.
\end{equation} 
Note that $ \Delta $ is relatively small for set-points as $ \dot{x}_2^*=0 $.  The standard LPV controller can still have comparable performance to other controllers. However, it fails to converge to dynamic references as $\Delta$ increases significantly with non-zero $\dot{x}_2^*$.

\begin{figure}[!bt]
	\centering
	\begin{tabular}{cc}
		\includegraphics[width=0.47\textwidth]{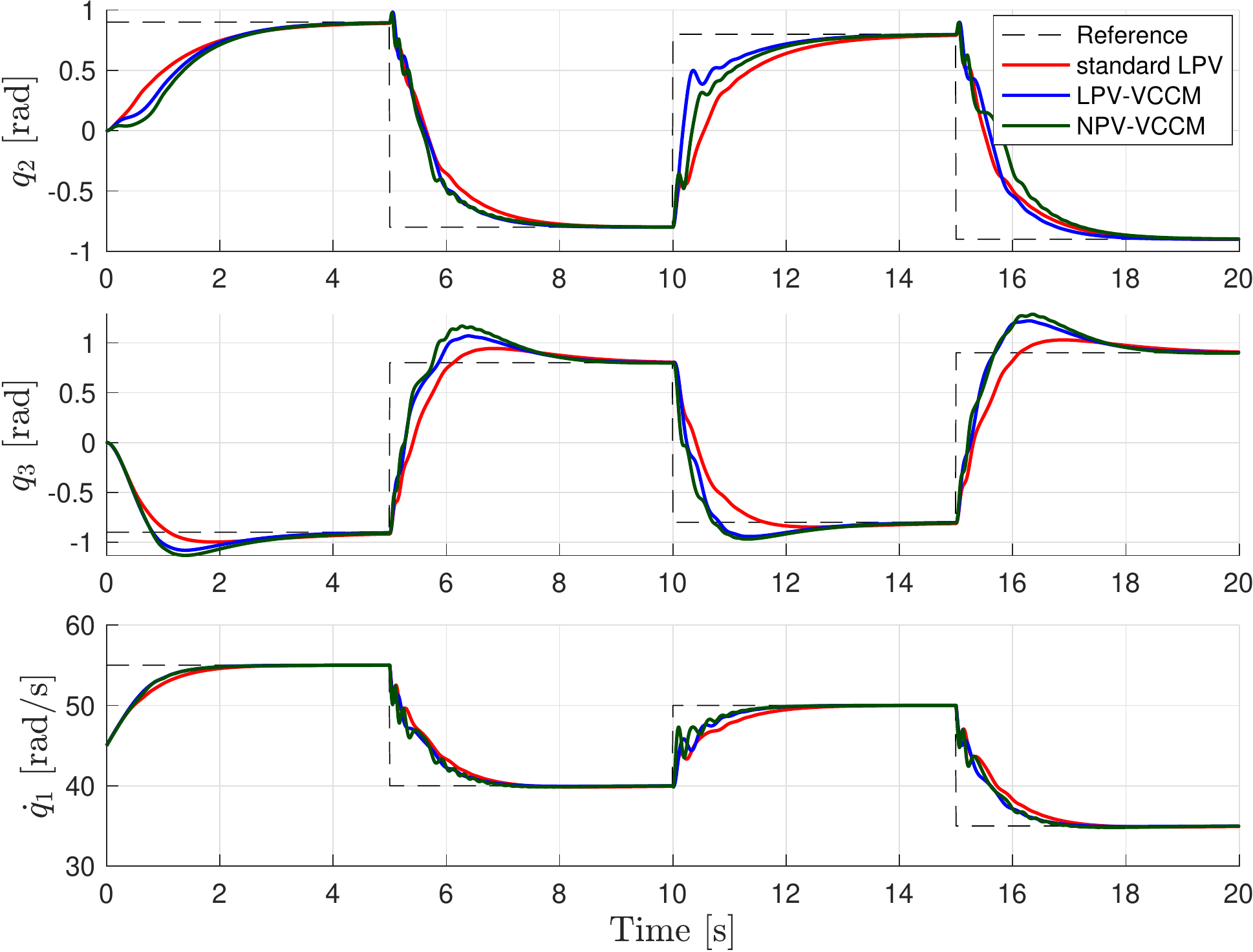} &
		\includegraphics[width=0.47\textwidth]{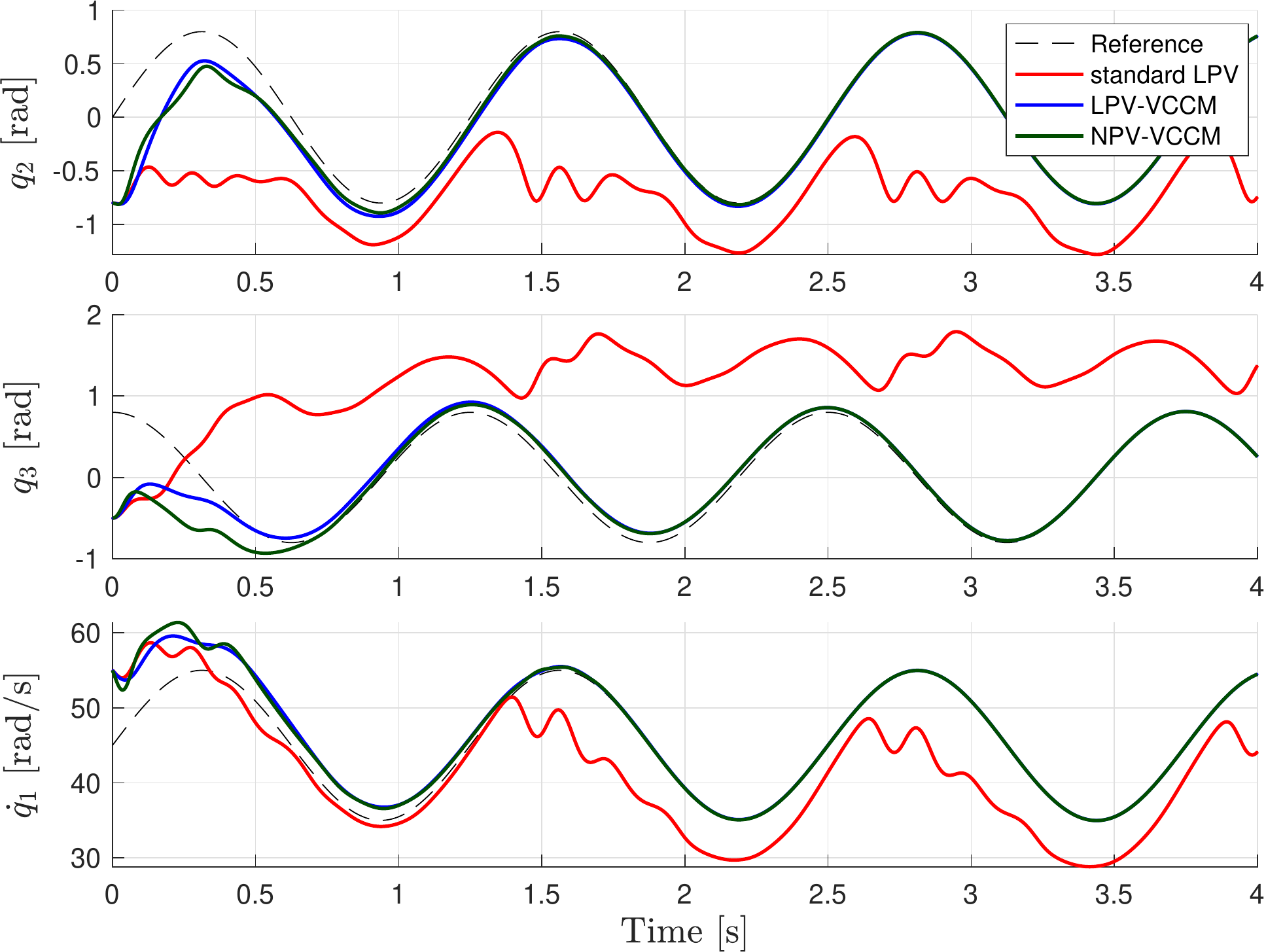} \\
		{\scriptsize(a) Set-point tracking} & {\scriptsize(b) Reference tracking}
	\end{tabular}
	\caption{OM-1 simulation results of different tracking tasks with controllers obtained from \eqref{eq:ccm-lmi}.}\label{fig:om1-response}
\end{figure}

We also compare the controllers from performance design for tracking of dynamic references. In the simulation, neither exogenous input disturbance nor model uncertainty is considered. Figure~\ref{fig:om1-performance}(a) depicts the response for a periodic reference with frequency of 0.8Hz under large initial error. The experimental test contains both input disturbance (i.e. friction) and various type of model uncertainties (e.g., unmodeled velocity filter and input saturation). The response for a periodic reference with frequency of 0.2Hz and small initial error is shown in Figure~\ref{fig:om1-performance}(b). Both the simulation and experimental results reveal that the LPV-VCCM and NPV-VCCM controllers have similar tracking performance while the standard LPV controller can fail to converge to the reference trajectory. This also can be seen from the performance comparison in Table~\ref{tab:om1-perform} where the standard LPV controller yields a much larger $J_T$ than LPV-VCCM and NPV-VCCM. Note that for the standard LPV control design better performance can possibly be obtained using different weights and/or a different controller structure. However, the guarantees of converging towards the reference trajectory are always absent, while for the VCCM based designs one \emph{does} have these guarantees.

\begin{table}[!bt]
	\centering
	\caption{Control performance comparison in OM-1: simulation - fast-varying reference and large initial error; experiment - slow-varying reference and small initial error.}\label{tab:om1-perform}
	\begin{tabular}{|c|c|c|c|c|}
		\hline
		& & & & \\[-2ex]
		Embedding & Gain bound $ \alpha $ & Realization & $ J_{T=4} $  (simulation) & $ J_{T=20} $  (experiment)\\  \hline 
		& & & & \\[-2ex]
		\multirow{2}{*}{LPV} & \multirow{2}{*}{0.4585}  & standard LPV & 446.3 & 32.8 \\ \cline{3-5}
		& & & & \\[-2ex]
		& & LPV-VCCM & 19.9 & 19.5 \\ \hline
		& & & & \\[-2ex]
		NPV & 0.4711 & NPV-VCCM & 19.3 & 7.7\\ \hline
	\end{tabular}
\end{table}

\begin{figure}[!bt]
	\centering
	\begin{tabular}{cc}
		\includegraphics[width=0.47\textwidth]{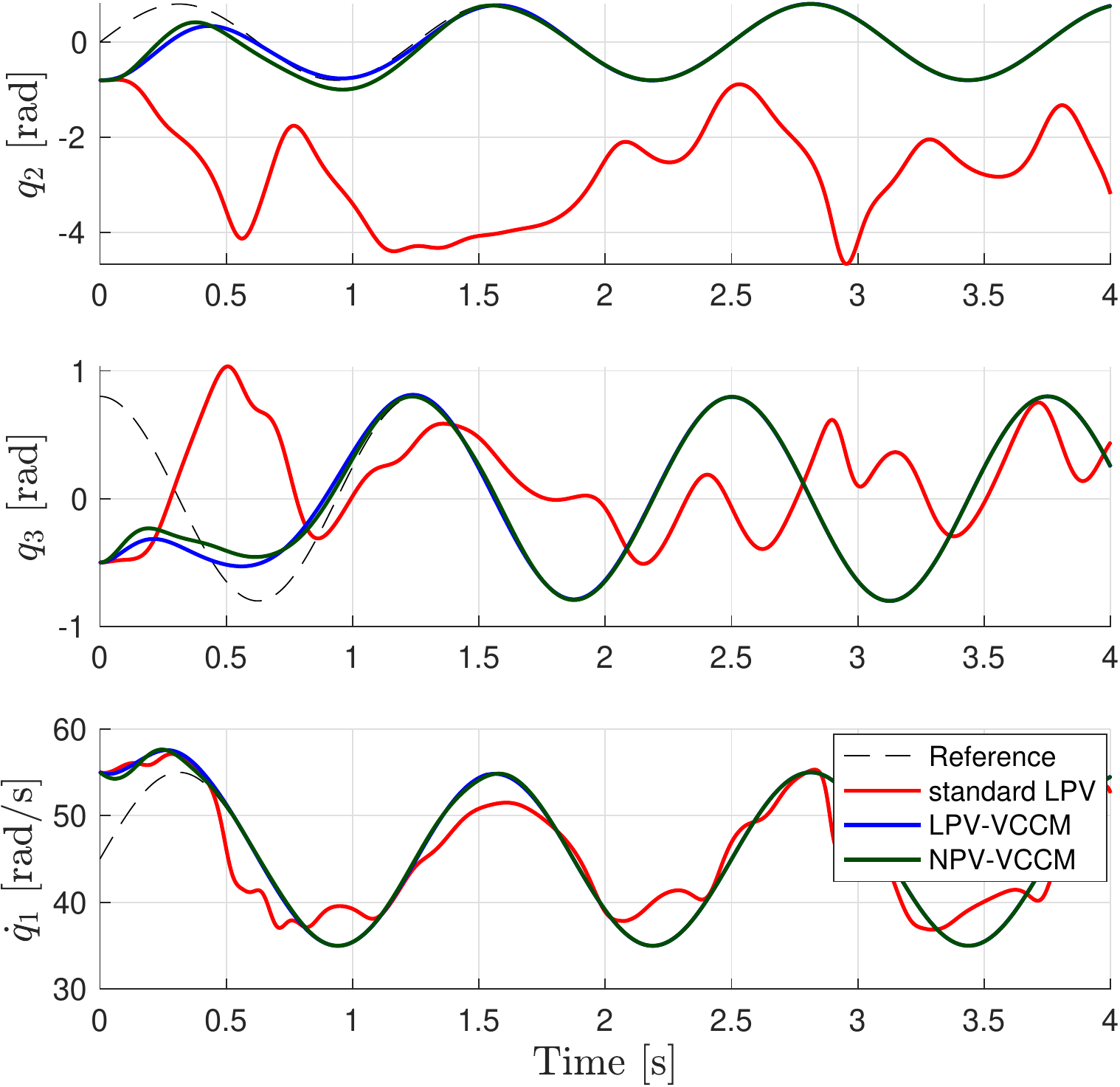} &
		\includegraphics[width=0.47\textwidth]{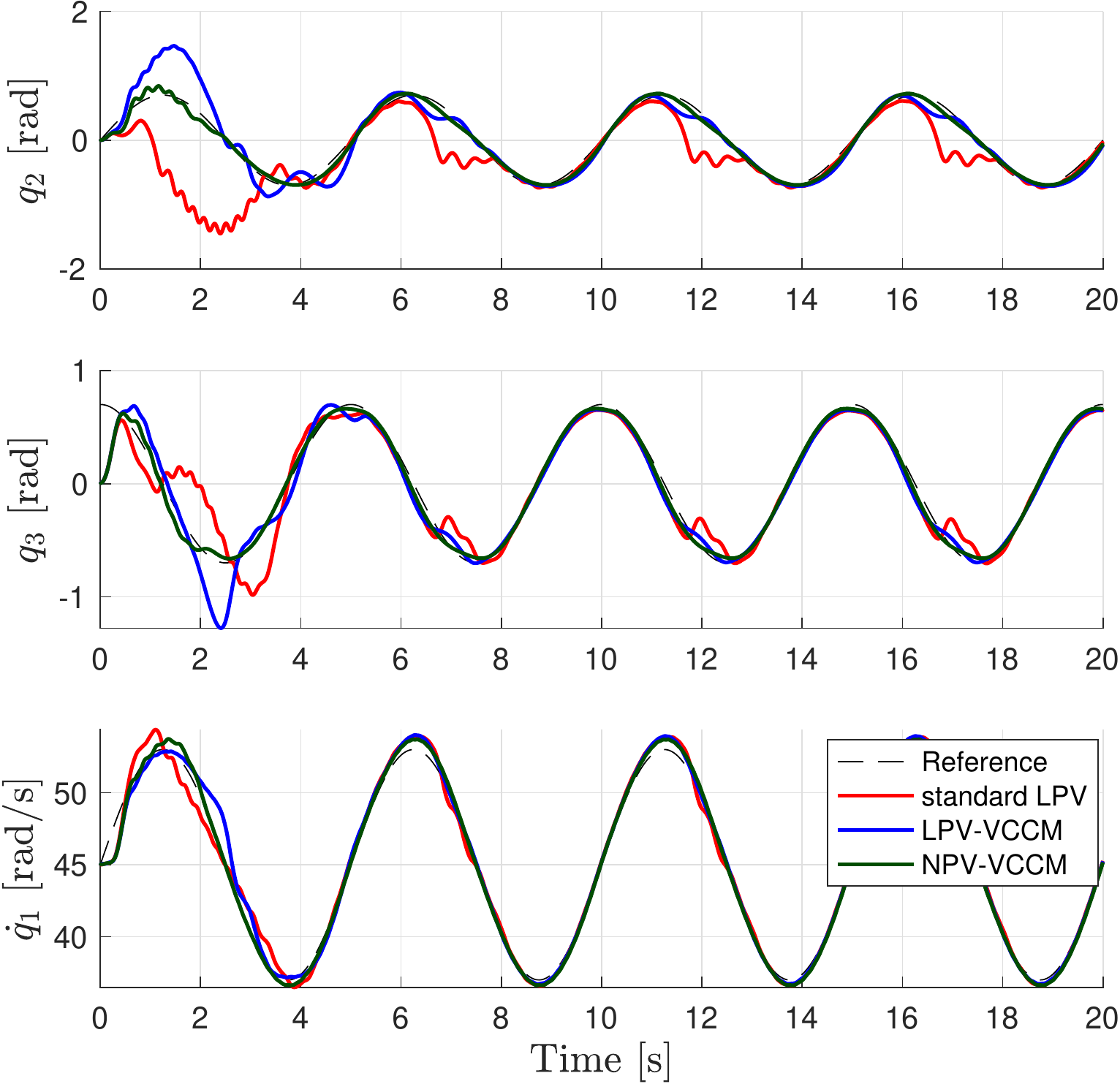} \\
		{\scriptsize (a) Simulation } & {\scriptsize (b) Experiment}
	\end{tabular}
	\caption{Comparison of controller obtained from \eqref{eq:rvccm-synthsis} for OM-1: simulation - fast-varying reference and large initial error; experiment - slow-varying reference and small initial error.}\label{fig:om1-performance}
\end{figure}

\subsection{Operating mode OM-2}\label{sec:result-om2}

We first simulate the CL responses of the standard LPV, LPV-VCCM and NPV-VCCM controllers obtained from Lyapunov design. Both small and large set-points are considered for gimbal $\mathrm{A}$, i.e., $ |q_4^*|= 0.36\pi $ and $ |q_4^*|=0.9\pi $. As shown in Figure~\ref{fig:om2-response}(a) where $|q_4^*|$ is small, due to the violation of Condition \textbf{C2}, the convergence speed of the standard LPV and LPV-VCCM controllers is slower than the NPV-VCCM approach. The performance deterioration of LPV-VCCM is less severe due to the specific choice of the feed-forward input \eqref{eq:om2-lpv-ff} where the residual term $ \Delta $ in \eqref{eq:residual} is minimized.

For moderate set-points $ |q_4^*|=0.45\pi $, the experimental result in Figure~\ref{fig:om2-response-exp} reveals a significant performance loss for standard LPV controller compared with the VCCM approach. This is due to the small stability margin of the standard LPV approach and the large uncertainties presented in the experimental setting (i.e., friction force, unmodeled velocity filter and input saturation).

When $|q_4^*|$ further increases, unstable CL behaviors are observed for both LPV-VCCM and NPV-VCCM controllers in simulation, as shown in Figure~\ref{fig:om2-response}(b). The main cause is that the variable $ q_2 $ exceeds the operation range for a large $|q_4^*|$. Those two controllers fail to keep $q_2$ within its operation range because the Lyapunov design does not take $ q_2 $ into account. Moreover, since the variables $q_2$ and $q_4$ are correlated, the LPV-VCCM and NPV-VCCM controllers give fast responses to $ q_4 $ by pushing $ q_2 $ towards the operation boundary, as shown in Figure~\ref{fig:om2-response}(a). 

\begin{figure}[!bt]
	\centering
	\begin{tabular}{cc}
		\includegraphics[width=0.47\textwidth]{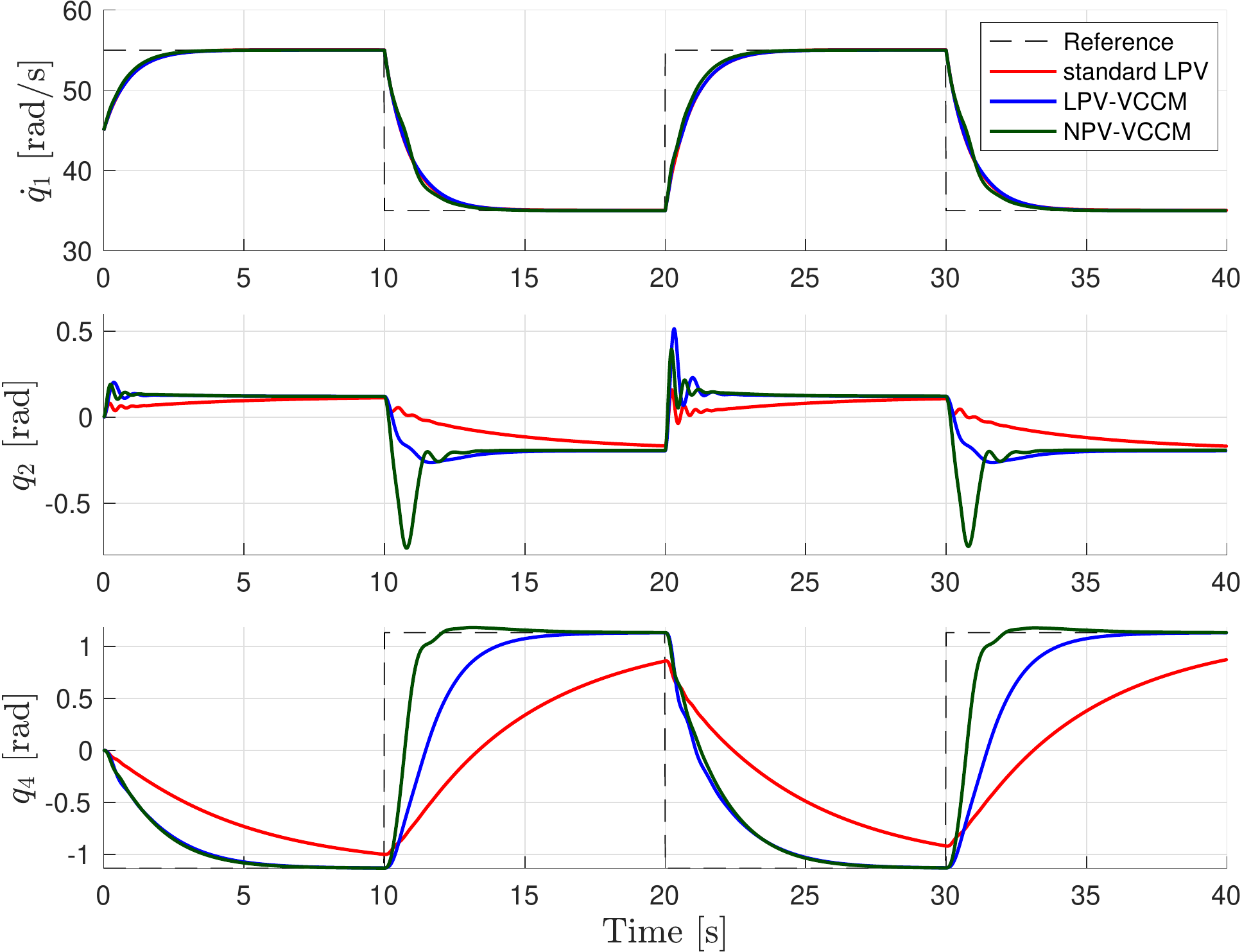} &
		\includegraphics[width=0.47\textwidth]{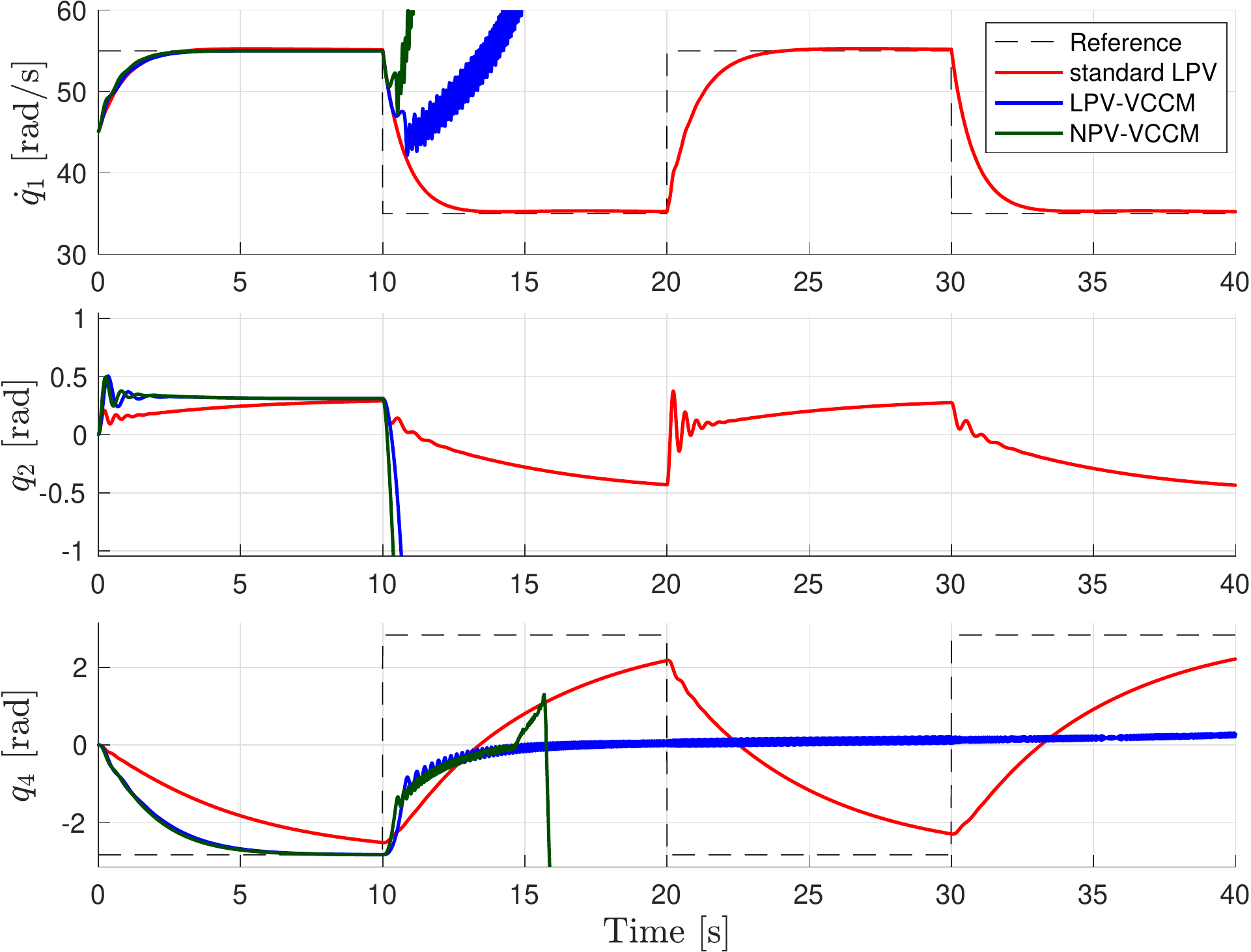} \\
		{\scriptsize(a) $ q_4^*=\pm0.36\pi $ } & {\scriptsize(b) $ q_4^*=\pm0.9\pi $}
	\end{tabular}
	\caption{OM-2 simulation results of different set-points with controllers obtained from \eqref{eq:ccm-lmi}.}\label{fig:om2-response}
\end{figure}

\begin{figure}[!bt]
	\centering
	\includegraphics[width=0.48\textwidth]{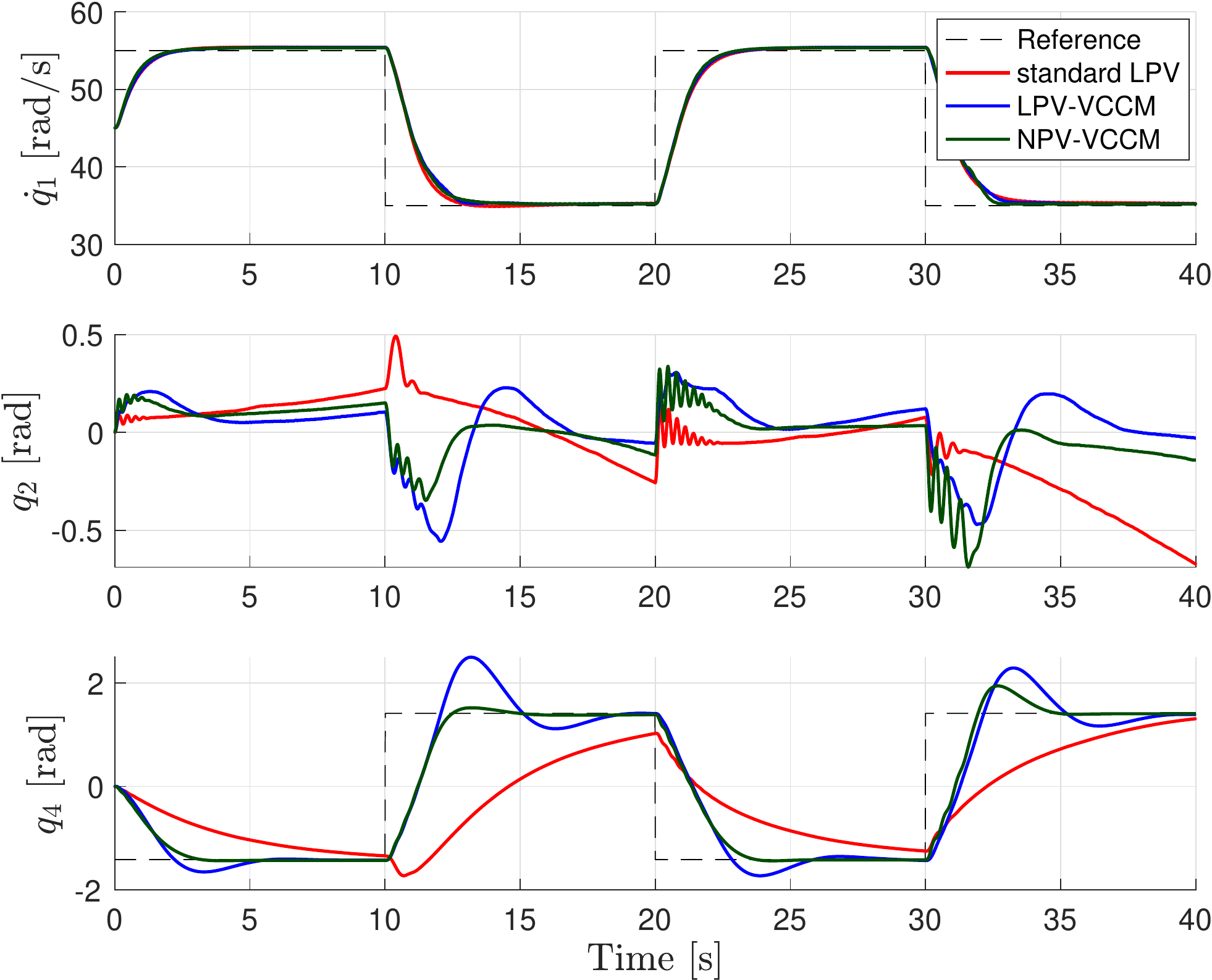} 
	\caption{OM-2 experimental result of controllers obtained from \eqref{eq:ccm-lmi} for moderate set-points.}\label{fig:om2-response-exp}
\end{figure}

Although the correlation between $q_2$ and $q_4$ causes stability issues, it also offers us a solution to address these issues via performance design. By choosing large weighting coefficients on $q_4$ and $\dot{q}_4$, it can help to keep $q_2$ within its operation range. We observe acceptable simulation and experimental responses for the choice of $W_1=\diag(5,0.1,1,4)$ and $W_2=\diag(20,10)$. Large $W_2$ is used to cope with the uncertainty from input saturation.

The synthesis results (Table~\ref{tab:om2-perform}) show that the universal $ \Lc_2 $-gain bounds for the LPV and NPV embedding are very close. Figure~\ref{fig:om2-performance}(a) shows that the standard LPV controller slightly outperforms the other two controllers in simulation (see also in Table~\ref{tab:om2-perform} where the standard LPV controller gives a smaller $J_T$). However, as shown in Figure~\ref{fig:om2-performance}(b), it leads to CL instability in the experiment where large uncertainties are presented. The loss of robustness is mainly due to the residual term in the Lyapunov analysis \eqref{eq:perf-loss} for set-point tracking. This term is caused by violation of Condition \textbf{C2} as analyzed in Section~\ref{sec:NPV-ff}. Although the LPV-VCCM controller also violates this condition, it is more robust than the stand LPV controller as it uses a feed-forward term \eqref{eq:om2-lpv-ff} that minimizes the residual term. For the NPV-VCCM controller which uses the same feed-forward term as the standard LPV approach, the difference is that it satisfies Condition \textbf{C2} due to the choice of the particular NPV embedding \eqref{eq:om2-npv}. Then, its robust stability and performance can be guaranteed by Theorem~\ref{thm:npv} and \eqref{thm:npv-performance} if $ q_2 $ is kept within its operation range.

\begin{table}[!bt]
	\centering
	\caption{Performance comparison with different embedding models and realizations for OM-2.}\label{tab:om2-perform}
	\begin{tabular}{|c|c|c|c|c|}
		\hline
		& & & &\\[-2ex]
		Embedding & Gain bound $ \alpha $ & Realization & $ J_{T=40} $ (simulation) & $ J_{T=40} $ (experiment) \\ \hline 
		& & & & \\[-2ex]
		\multirow{2}{*}{LPV} & \multirow{2}{*}{1.2035}  & standard LPV & \num{2.1314e+4} & unstable\\ \cline{3-5}
		& & & & \\[-2ex]
		& & LPV-VCCM & \num{2.6816e4} & \num{1.0901e5}\\ \hline
		& & &  &\\[-2ex]
		NPV & 1.2094 & NPV-VCCM & \num{2.8506e4} & \num{1.2753e5} \\ \hline
	\end{tabular}
\end{table}

\begin{figure}[!bt]
	\centering
	\begin{tabular}{cc}
		\includegraphics[width=0.47\textwidth]{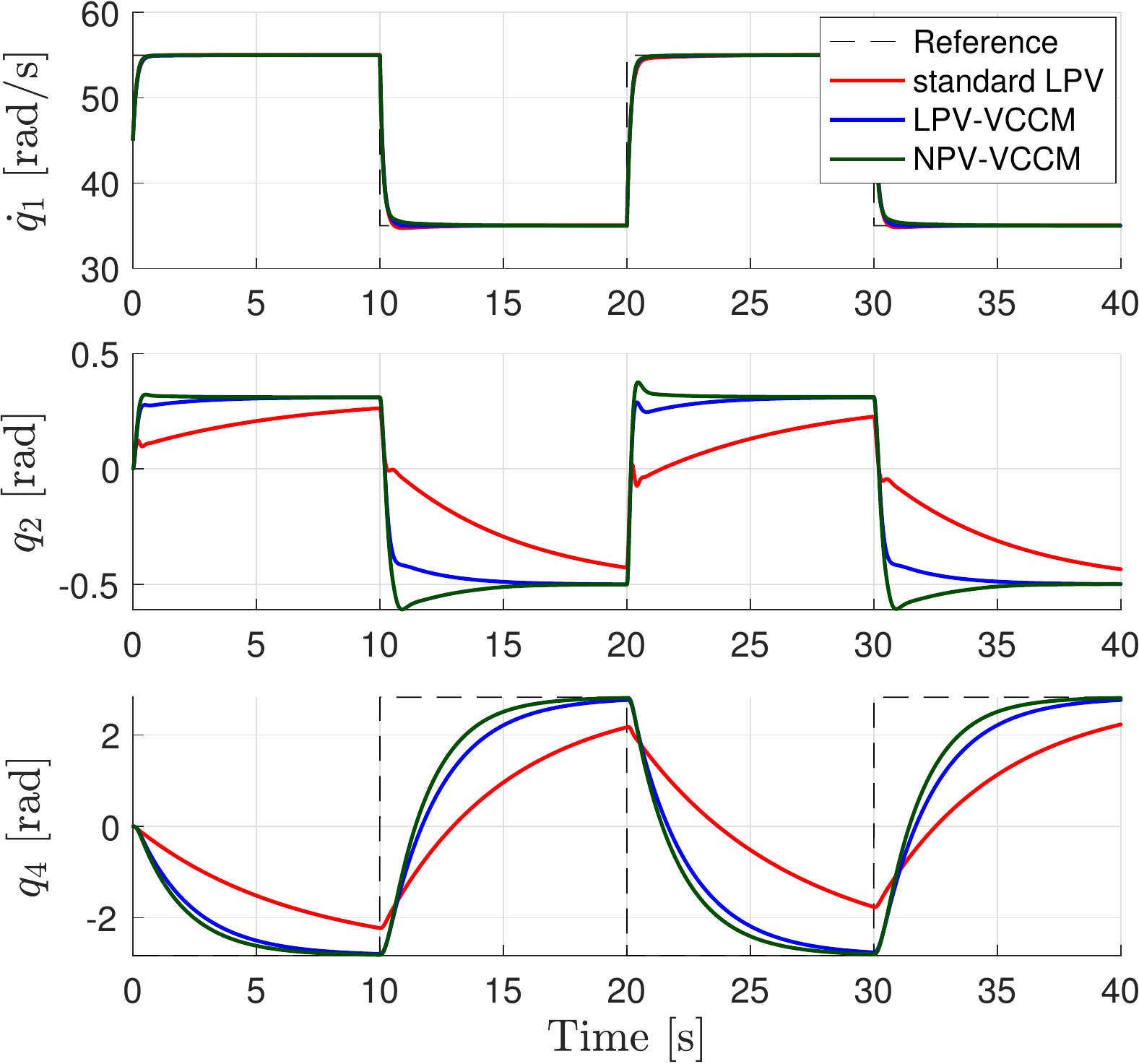} &
		\includegraphics[width=0.47\textwidth]{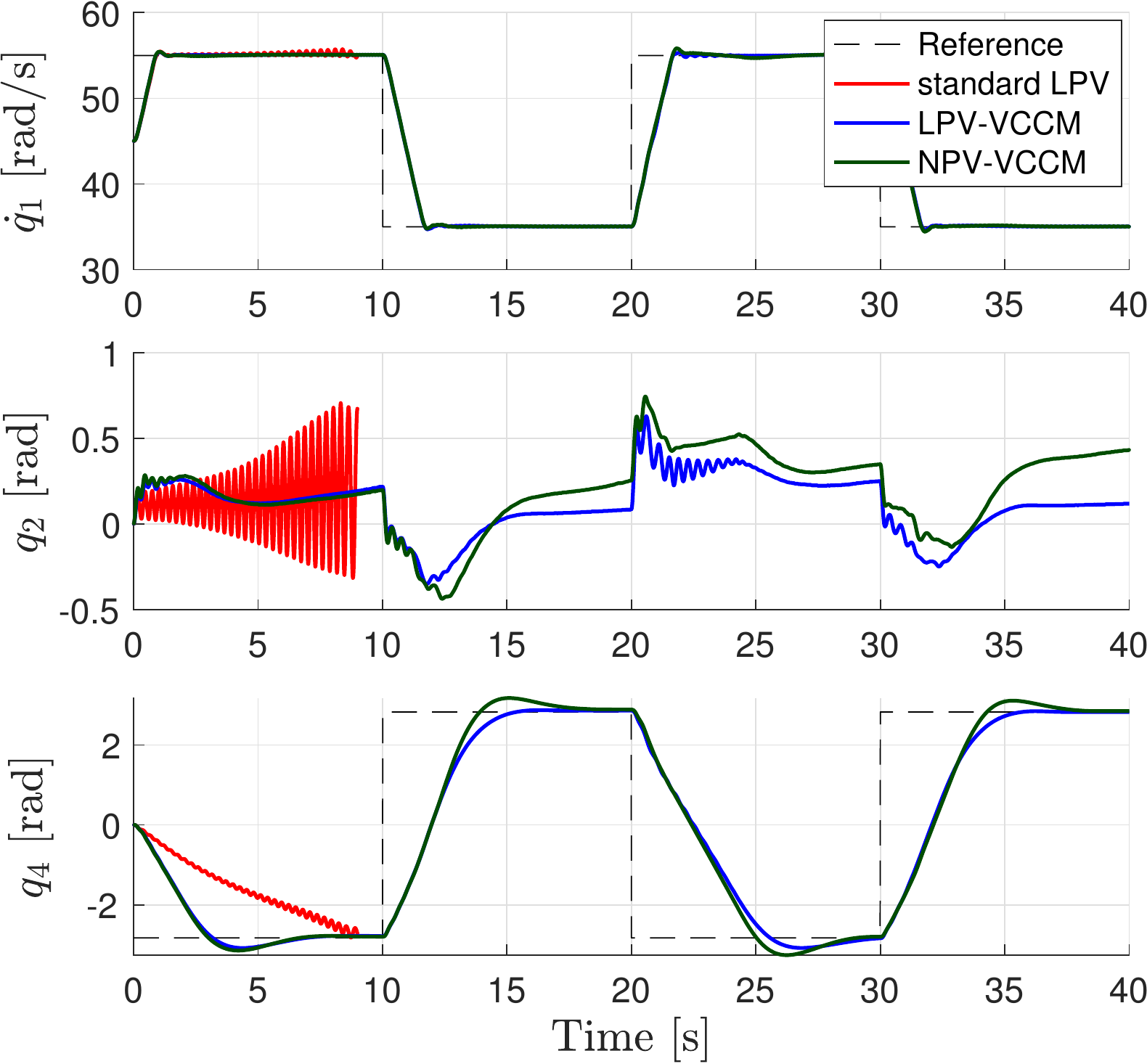} \\
		{\scriptsize (a) Simulation} & {\scriptsize (b) Experiment}
	\end{tabular}
	\caption{OM-2 set-point tracking comparison of controllers obtained from \eqref{eq:rvccm-synthsis}.}\label{fig:om2-performance}
\end{figure}

\section{Conclusion}\label{sec:conclusion}
In this paper, we applied a virtual control contraction metric (VCCM) based nonlinear parameter-varying (NPV) approach to design state-feedback tracking controllers for two (both fully- and under-actuated) operation modes of a control moment gyroscope. This approach includes three steps: 1) choose a NPV embedding, 2) search for a VCCM and 3) realize the controller back into the original state/input space. Since the NPV embedding is non-unique, but essential for control synthesis and realization, we provided two conditions for the choice of NPV models: (\textbf{C1}) the NPV virtual system is universally stabilizable, and (\textbf{C2}) the desired reference is an admissible state trajectory of the NPV embedded system. Condition \textbf{C1} can be easily verified via a convex control synthesis formulation similar to the conventional LPV approach. Condition \textbf{C2} determines whether a control realization can provide closed-loop stability and performance guarantees or not.  As shown in simulation and experimental comparisons, the standard LPV control realization does not ensure closed-loop stability and performance as it does not satisfy Condition \textbf{C2} while the VCCM based realization addresses this issue by utilizing additional freedom in feed-forward control design. Future works include dynamic controller realizations to relax Condition \textbf{C2}.

%\bibliographystyle{IEEEtran}
%\bibliography{ref}


\end{document}